%% file: main.tex
\documentclass[acmsmall,screen,nonacm]{acmart}
\AtBeginDocument{%
  }

\usepackage{stmaryrd}
\usepackage{mathtools}
\mathtoolsset{showonlyrefs=true}
\usepackage{tikz}
\usepackage{enumitem}
\usepackage{subcaption}

\newif\iflong
\longtrue

\newcommand{\interpret}[1]{\llbracket #1 \rrbracket}
\newcommand{\arity}{\mathsf{ar}}
\newcommand{\sort}{\mathsf{sort}}
\newcommand{\propsort}{o}
\newcommand{\True}{\mathtt{true}}
\newcommand{\False}{\mathtt{false}}
\newcommand{\FreeTermVars}{\mathsf{ftv}}
\newcommand{\FreePredVars}{\mathsf{fpv}}
\newcommand{\MuVal}{\textsc{MuVal}}
\newcommand{\MuValPPM}{\textsc{MuValPPM}}
\newcommand{\MuStrat}{\textsc{MuStrat}}
\newcommand{\PCSat}{\textsc{PCSat}}
\newcommand{\Range}[2]{\{ #1, \dots, #2 \}}

\newcommand{\Play}[2][]{\mathbf{Play}_{#1}(#2)}
\newcommand{\MaximalPlay}[2][]{\mathbf{MaxPlay}_{#1}(#2)}
\newcommand{\StrategyPlay}[3][]{\mathbf{Play}_{#1}^{#2}(#3)}
\newcommand{\GameSemantics}[1]{G_{#1}}
\newcommand{\WellFoundedGame}[1]{G^{\mathrm{wf}}_{#1}}
\newcommand{\DisjunctivelyWellFoundedGame}[1]{G^{\mathrm{dwf}}_{#1}}
\newcommand{\CounterGame}[1]{G^{\mathrm{cnt}}_{#1}}

\AtEndPreamble{%
\theoremstyle{acmdefinition}

\newtheorem{remark}[theorem]{Remark}}

\usepackage[draft,anonymous,warnconjecture]{macros.aux}

\begin{document}
\title{Solving First-Order Fixed-Point Logics via a Least-to-Greatest Transformation Based on Game Semantics}

\author{Satoshi Kura}
\email{satoshikura@acm.org}
\orcid{0000-0002-3954-8255}
\affiliation{%
  \institution{Waseda University}
  \city{Tokyo}
  \country{Japan}
}

\author{Hiroshi Unno}
\email{hiroshi.unno@acm.org}
\orcid{0000-0002-4225-8195}
\affiliation{%
  \institution{Tohoku University}
  \city{Sendai}
  \country{Japan}
}

\begin{abstract}
Fixed-point logics provide an expressive intermediate framework for reasoning about temporal properties of programs.
One of the key approaches to solving their validity checking problem is via transformations from least fixed points to greatest fixed points (\emph{$\mu$-to-$\nu$ transformations}), which generalizes a reduction from termination verification to safety verification studied in binary reachability analysis.
In this paper, we introduce game-semantic interpretations of $\mu$-to-$\nu$ transformations.
We first introduce a new $\mu$-to-$\nu$ transformation based on parity relations.
We show that solving $\mu$-to-$\nu$-transformed fixed-point equation systems corresponds to finding winning strategies in the game semantics of the original fixed-point equation systems.
We apply the same game-semantic framework to interpret two existing $\mu$-to-$\nu$ transformations, one by Kobayashi et al.\ and the other by Unno et al, and show that they admit analogous game-semantic interpretations.
Furthermore, we show that the game introduced by Tsukada et al.\ corresponds to an alternative characterization of the winning condition.
On the implementation side, we propose optimization techniques for efficiently solving our new $\mu$-to-$\nu$ transformation.
We implement these techniques in a fixed-point logic solver, compare our approach with existing solvers, and demonstrate the effectiveness of the proposed optimizations through experiments.
\end{abstract}


\maketitle

\input{intro}
\input{prelim}

\input{ppm-translation}

\input{popl23}
\input{sas19}
\input{popl25}

\input{implementation}
\input{experiment}
\input{related-work}

\section{Conclusions and Future Work}
\label{sec:conclusion}
We have presented a novel $\mu$-to-$\nu$ transformation for $\mu$CLPs based on parity relations.
We proved that the transformation is sound and complete and established its game-semantic interpretation.
We also showed that the existing $\mu$-to-$\nu$ transformations~\cite{KobayashiSAS2019,UnnoPOPL2023} and the game introduced in~\cite{TsukadaPOPL2025} can be understood from a unified game-semantic perspective.
Finally, we implemented our transformation together with optimization techniques, including SCC-wise parity relations and priority reassignment, and demonstrated their effectiveness through experiments on benchmark suites.
As future work, we plan to develop advanced methods for exchanging information between primal and dual $\mu$CLP solving.

\checkpagelimit{25}

\bibliographystyle{ACM-Reference-Format}
\bibliography{ref}

\iflong
\clearpage
\appendix
\input{complete-progress-measure}
\input{proofs}
\fi

\end{document}

%% file: intro.tex
\section{Introduction}

Fixed-point logics have been widely studied as a foundation for reasoning about temporal properties of programs.
While much of the early work focused on finite-state systems, such as those captured by the modal $\mu$-calculus, recent research has addressed verification of infinite-state programs~\cite{NanjoLICS2018,UnnoCAV2021,UnnoPOPL2023,KobayashiSAS2019,KobayashiESOP2018,CathcartBurnPOPL2018,BeyeneCAV2013,BjornerFieldsofLogicandComputationII2015,UnnoPOPL2013,DeAngelisTheoryandPracticeofLogicProgramming2022}.
For example, properties such as the weakest precondition and weakest liberal precondition of a loop can be characterized as the least and greatest fixed points, respectively, and are thus naturally expressible within a fixed-point logic~\cite{BlassComputationTheoryandLogic1987}.
Numerous program verification methods have been developed through reductions to solving constrained Horn clauses (CHCs)~\cite{BeyeneCAV2013,BjornerFieldsofLogicandComputationII2015,CathcartBurnPOPL2018,UnnoPOPL2013,DeAngelisTheoryandPracticeofLogicProgramming2022}, which can be seen as a fragment of a first-order fixed-point logic without $\mu$ and $\exists$.
Allowing alternation of least and greatest fixed points further increases the expressiveness of fixed-point logics.
For instance, model checking of infinite-state while-programs against modal $\mu$-calculus specifications can be reduced to the validity checking problem for a first-order fixed-point logic with fixed-point alternation~\cite{KobayashiSAS2019}.
Other applications of fixed-point logics include game solving and reactive synthesis~\cite{SamuelFSE2021,MaderbacherFMCAD2022,SchmuckCAV2024,HeimPOPL2024,HeimPOPL2025,HeimCAV2025}; implementability of global protocols~\cite{LiOOPSLA2025a,LiCAV2025}; and semantics-guided synthesis~\cite{MurphyPLDI2025}.
These various applications highlight the role of fixed-point logics as a powerful intermediate formalism for program verification and synthesis.
Thus, efficient solvers for fixed-point logics are a key component towards practical verification and synthesis tools.

The validity checking problem of fixed-point logics can be solved by applying transformations from least fixed points to greatest fixed points \cite{KobayashiSAS2019,UnnoPOPL2023}.
We call such transformations \emph{least-to-greatest fixed-point transformations} or simply \emph{$\mu$-to-$\nu$ transformations}.
The idea of such $\mu$-to-$\nu$ transformations was inspired by, and generalizes, the reduction of termination verification problems to safety verification problems done in \emph{binary reachability analysis} \cite{PodelskiLICS2004,CookPLDI2006,GrebenshchikovPLDI2012,KuwaharaESOP2014,FedyukovichCAV2018}.
The resulting safety verification problems can then be solved by synthesizing suitable invariants.

To give some intuition, we first consider the simple setting where least and greatest fixed points do not alternate.
As a simple example, consider the termination of the following program.
\[ \texttt{while}\ (x \neq 0)\ \{\ x \coloneqq x - 1;\ \} \]
This problem can be encoded by the following least fixed-point equation.
\[ X(x) \quad=_{\mu}\quad x = 0 \lor X(x - 1) \qquad\qquad \text{where $x \in \mathbb{Z}$} \]
The least fixed point $X^{*}$ of the above equation gives the set of terminating states, which in this case is the set of non-negative integers.
To verify that the program terminates from the initial state, say $x = 100$, we need to show $100 \in X^{*}$.
This amounts to finding a \emph{lower bound} $X' \subseteq X^{*}$ such that $X'(100)$ holds.

Finding lower bounds of greatest fixed points (e.g.\ safety verification) is relatively easy, as it can be reduced to the problem of finding a \emph{post-fixed point} by the Knaster--Tarski theorem.
On the other hand, finding lower bounds of least fixed points (e.g.\ termination verification) is more challenging, as we cannot directly apply the Knaster--Tarski theorem.
One way to address this is through $\mu$-to-$\nu$ transformations.
In the current example, a $\mu$-to-$\nu$ transformation allows us to consider lower bounds of the following greatest fixed-point equation instead of the original least fixed-point equation.
\begin{equation}
	X'(x) \quad=_{\nu}\quad x = 0 \lor (X'(x - 1) \land R(x, x - 1)) \qquad \text{$R$: well-founded}
	\label{eq:example-mu-to-nu}
\end{equation}
Here, $R$ is an arbitrary well-founded relation.
Then, the greatest fixed point $X'^{*}$ of the above equation is a lower bound of $X^{*}$ because the well-foundedness of $R$ prohibits infinite unfolding of the equation\footnote{In general, $X'^{*}$ can be a proper subset of $X^{*}$ depending on the choice of $R$, but it is known that there exists a well-founded relation $R$ such that $X'^{*} = X^{*}$.}.
By the Knaster--Tarski theorem, any post-fixed point is a lower bound of the greatest fixed point.
Thus, the validity of $X(100)$ can be reduced to the problem of finding a predicate $X'$ and a well-founded relation $R$ that satisfy the following constraints.
\[ X'(100), \qquad X'(x) \implies x = 0 \lor (X'(x - 1) \land R(x, x - 1)) \]
In other words, the termination verification problem is now reduced to a safety verification problem.
The above constraints require finding a suitable invariant $X'$ such that (a) the initial state $x = 100$ is included in $X'$, (b) $X'$ is closed under the transition relation of the program, and (c) the transition relation is well-founded on $X'$.

By the above reduction, we obtain constrained Horn clauses (CHCs)\footnote{Precisely, we obtain \emph{co-CHCs}, which are equivalent to CHCs via mutual translations.}, which can be solved by off-the-shelf CHC solvers (if $R$ is given).
Such reductions have been studied in the context of binary reachability analysis and then generalized to more general first-order fixed-point logics in which least and greatest fixed points can alternate.
\citet{KobayashiSAS2019} proposed a $\mu$-to-$\nu$ transformation for first-order fixed-point logics, but it is not complete.
\citet{UnnoPOPL2023} proposed another transformation, which achieves completeness.
However, these transformations

In this paper, we propose a novel and simple $\mu$-to-$\nu$ transformation for first-order fixed-point logics, together with its interpretation in terms of game semantics of fixed-point logics.
Specifically, we first introduce a notion of \emph{parity progress relation} (or \emph{parity relation} for short), which is a generalization of the notion of \emph{well-founded relation} and can be constructed from parity progress measures~\cite{JurdzinskiSTACS2000}.
Then, the notion of parity relations allows us to generalize the transformation in~\eqref{eq:example-mu-to-nu}, which leads to a sound and complete reduction of the validity checking problem of first-order fixed-point logics to a constraint solving problem.

Our $\mu$-to-$\nu$ transformation can be naturally interpreted in terms of game semantics of fixed-point logics~\cite{BradfieldHandbookofModelChecking2018,BaldanPOPL2019,BaldanCONCUR2020}, which interprets a fixed-point equation system as a parity game such that the validity of the equation system corresponds to the existence of a winning strategy.
Similarly to the example above, our $\mu$-to-$\nu$ transformation inserts parity relations into a given fixed-point equation system, which ensures that the parity condition is satisfied for any infinite play in the corresponding parity game.
Therefore, if there exists a solution for the constraint solving problem obtained by our transformation, then there exists a memoryless winning strategy in the parity game.

Game semantics also provides a unified perspective on the $\mu$-to-$\nu$ transformations of~\cite{KobayashiSAS2019,UnnoPOPL2023} and the game introduced in~\cite{TsukadaPOPL2025}.
Specifically, these works can be uniformly interpreted as equipping the original parity game with auxiliary partial information about the history and then rephrasing the winning condition in terms of the auxiliary information.
By characterizing the winning condition in terms of well-founded relations, we obtain a game-semantic interpretation of~\cite{UnnoPOPL2023}.
Solving the $\mu$-to-$\nu$-transformed equation system of~\cite{UnnoPOPL2023} corresponds to constructing a winning strategy in the augmented parity game.
Similarly, the game introduced in~\cite{TsukadaPOPL2025} can be viewed as a reformulation of the winning condition in terms of disjunctively well-founded relations.
Likewise, the transformation in~\cite{KobayashiSAS2019} can be understood in terms of a game augmented with counters that bound the number of unfoldings of least fixed points.
This game-semantic analysis also reveals that the original game of~\cite{TsukadaPOPL2025} imposes an unnecessary restriction on the winning condition.

We implement our new $\mu$-to-$\nu$ transformation on top of the existing solver \MuVal~\cite{UnnoPOPL2023} for first-order fixed-point logics.
Our implementation generates constraints over predicate variables that involve (unknown) parity relations.
To synthesize parity relations, we provide templates for parity relations and implement them in the existing constraint solver \PCSat~\cite{UnnoCAV2021}.
We further develop several optimization techniques based on dependency analysis, which reduce the size of the generated constraints.
We compared our implementation with the existing solver \MuVal~\cite{UnnoPOPL2023} and \MuStrat~\cite{TsukadaPOPL2025} on benchmarks from the literature.
We also evaluated the effectiveness of our optimization techniques.

Our contributions are summarized as follows.
\begin{itemize}
	\item We present a novel and simple $\mu$-to-$\nu$ transformation for first-order fixed-point equation systems.
	We show that our $\mu$-to-$\nu$ transformation admits a natural interpretation in terms of game semantics for fixed-point logics~\cite{BaldanPOPL2019,BaldanCONCUR2020} and yields a sound and complete reduction of the validity checking problem for first-order fixed-point logics to a constraint solving problem.
	\item We discuss how the existing methods~\cite{KobayashiSAS2019,UnnoPOPL2023,TsukadaPOPL2025} can be understood from a uniform game-semantic perspective.
	In particular, we interpret the transformation proposed in~\cite{UnnoPOPL2023} in terms of well-founded relations, and the game introduced in~\cite{TsukadaPOPL2025} in terms of disjunctively well-founded relations.
	We also show that the transformation proposed in~\cite{KobayashiSAS2019} admits a game-semantic interpretation using counters.
	This game-theoretic perspective provides a unified view of these methods and our own.
	\item We implement our method and several optimization techniques based on dependency analysis to reduce generated constraints.
	We evaluate effectiveness of our implementation and optimization techniques via experiments.
\end{itemize}

\paragraph{Outline.}
This paper is organized as follows.
Section~\ref{sec:preliminaries} introduces first-order fixed-point equation systems, called $\mu$CLPs~\cite{UnnoPOPL2023}, together with the necessary background on parity games and the game semantics of $\mu$CLPs.
Section~\ref{sec:validity-checking} presents a novel, sound, and complete $\mu$-to-$\nu$ transformation for $\mu$CLPs based on parity relations.
Section~\ref{sec:comparison-popl23} relates the game semantics of $\mu$CLPs to the transformation of~\cite{UnnoPOPL2023} via well-founded relations.
Section~\ref{sec:comparison-sas19} interprets the transformation of~\cite{KobayashiSAS2019} in terms of game semantics with counters.
Section~\ref{sec:comparison-popl25} relates the game semantics of $\mu$CLPs to the game introduced in~\cite{TsukadaPOPL2025} via disjunctively well-founded relations.
Section~\ref{sec:implementation} describes efficient reduction techniques from $\mu$CLPs to constraint solving problems.
Section~\ref{sec:experiments} presents experimental results.
Finally, Section~\ref{sec:related-work} discusses related work, and Section~\ref{sec:conclusion} concludes the paper.

%% file: prelim.tex
\section{Preliminaries}\label{sec:preliminaries}

\subsection{$\mu$CLP: First-Order Fixed-Point Equation Systems}

Following~\cite{UnnoPOPL2023}, we define $\mu$CLPs, which are fixed-point equation systems over first-order predicate variables.
Then, we define the validity checking problem, which is the target problem of this paper.

\begin{definition}
	\emph{Terms} $t$ and \emph{(first-order) formulas} $\phi$ are defined as follows.
	\begin{align}
		t \ &\coloneqq\ x \mid f(t_1, \dots, t_{\arity(f)}) \\
		\phi, \psi\  &\coloneqq\ X(t_1, \dots, t_{\arity(X)}) \mid p(t_1, \dots, t_{\arity(p)}) \mid \lnot \phi \mid \phi \land \psi\mid \phi \lor \psi \mid \forall x. \phi \mid \exists x. \phi
	\end{align}
	Here, $x$ is a term variable, $f$ is a function symbol, $X \in \mathcal{X}$ is a predicate variable, and $p$ is a predicate symbol.
	We write $\arity(f), \arity(X), \arity(p) \in \mathbb{N}$ for their arities and $\sort(f), \sort(X), \sort(p)$ for their sorts.
	For each function symbol $f$, the sort is of the form $\sort(f) = \tilde{s} \to s'$ where $\tilde{s} = (s_1, \dots, s_{\arity(f)})$ is a list of sorts for the arguments of $f$ and $s'$ is the sort of the result.
	Similarly, for each predicate variable $X$ (predicate symbol $p$), the sort is of the form $\sort(X) = \tilde{s} \to \propsort$ ($\sort(p) = \tilde{s} \to \propsort$) where $\propsort$ is the sort of propositions (i.e., boolean values).
\end{definition}
We always assume that each term variable $x$ is associated with a sort $\sort(x)$ and that function symbols, predicate variables, and predicate symbols are applied to terms of appropriate sorts.
The set of free term variables and free predicate variables in $\phi$ is denoted by $\FreeTermVars(\phi)$ and $\FreePredVars(\phi)$, respectively.
We write $\phi[t/x]$ for the formula obtained by substituting a term $t$ for a term variable $x$ in $\phi$, and $\phi[\lambda (y_1, \dots, y_{\arity(X)}). \psi/X]$ for the formula obtained by substituting a formula $\psi[t_1/y_1, \dots, t_{\arity(X)}/y_{\arity(X)}]$ for each occurrence of $X(t_1, \dots, t_{\arity(X)})$ in $\phi$.

\begin{definition}
	A \emph{$\mu$}CLP is a list of equations of the following form.
	\begin{equation}
		X_1(\tilde{x}_1) =_{\eta_1} \phi_1; \qquad \dots; \qquad X_n(\tilde{x}_n) =_{\eta_n} \phi_n
		\label{eq:equation-system}
	\end{equation}
	Here, for each $i$, (a) $\tilde{x}_i = (x_1, \dots, x_{\arity(X_i)})$ is a list of term variables, (b) $\phi_i$ is a first-order formula such that every occurrence of predicate variables in $\phi_i$ is positive, and (c) $\eta_i \in \{ \mu, \nu \}$ is a symbol indicating either least or greatest fixed point.
	Note that the positive-occurrence condition is necessary to ensure the monotonicity of $\phi_i$ with respect to each predicate variable.
	If $\eta_i = \nu$ (resp.\ $\eta_i = \mu$), then we call the $i$-th equation a \emph{$\nu$-equation} (resp.\ \emph{$\mu$-equation}).
	In this paper, we always assume that $\mu$CLPs are \emph{closed}, i.e., $\FreeTermVars(\phi_i) \subseteq \tilde{x}_i$ and $\FreePredVars(\phi_i) \subseteq \{ X_1, \dots, X_n \}$ for each $i$.
\end{definition}
A $\mu$CLP~\eqref{eq:equation-system} defines for each predicate variable $X_i$, a set $X_i^{*}$ of values of sort $\tilde{s}_i$ where $\sort(X_i) = \tilde{s}_i \to \propsort$.
We call $(X_1^{*}, \dots, X_n^{*})$ the \emph{solution} of the $\mu$CLP.
The formal definition (Definition~\ref{def:solution}) will be explained later in Section~\ref{sec:semantics-muCLP}.

\begin{definition}\label{def:muCLP-with-query}
	A \emph{query} for a $\mu$CLP~\eqref{eq:equation-system} is a first-order formula $\psi$ such that $\FreePredVars(\psi) \subseteq \{ X_1, \dots, X_n \}$.
	A query $\psi$ is \emph{closed} if $\FreeTermVars(\psi) = \emptyset$.
	Given a pair of a $\mu$CLP $E$ and a closed query $\psi$, we say $(E, \psi)$ is \emph{valid}, written $\models (E, \psi)$, if the solution of $E$ satisfies $\psi$.
\end{definition}

$\mu$CLPs with a query can be seen as an alternative presentation of formulas of a first-order fixed-point logic with nested fixed point operators~\cite{BradfieldCSL1999,KobayashiSAS2019}.
Below, we consider a first-order fixed-point logic with the least and greatest fixed-point operators, written as $\mu X. \lambda \tilde{x}. \phi$ and $\nu X. \lambda \tilde{x}. \phi$, respectively, where $\sort(X) = \tilde{s} \to \propsort$ and $\sort(\tilde{x}) = \tilde{s}$.
We can inductively define a transformation $\mathrm{toFFL}$ from $\mu$CLPs to such nested fixed-points as follows.
\begin{align}
	&\mathrm{toFFL}(\emptyset, \psi) &&\coloneqq &&\psi \\\
	&\mathrm{toFFL}(E'; X(\tilde{x}) =_{\eta} \phi,\ \psi) &&\coloneqq &&\mathrm{toFFL}(E'[\eta X. \lambda \tilde{x}. \phi/X],\ \psi[\eta X. \lambda \tilde{x}. \phi/X])
\end{align}
Via the transformation $\mathrm{toFFL}$, the rightmost equation in the $\mu$CLP corresponds to the innermost fixed-point operator in the first-order fixed-point logic.
Hence, the solution of a $\mu$CLP depends on the order of equations.
The definition of the solution of $\mu$CLPs, which we will define in Section~\ref{sec:semantics-muCLP}, is consistent with this transformation.

\begin{example}\label{ex:muCLP-to-FFL}
	Consider the following $\mu$CLP over predicate variables $X_1 : \mathtt{int} \to \propsort$ and $X_2 : (\mathtt{int}, \mathtt{int}) \to \propsort$.
	\begin{equation}
		X_1(x) =_{\nu} X_1(x + 1) \land X_2(x, 0); \qquad X_2(x, y) =_{\mu} x = y \lor X_2(x, y + 1)
		\label{eq:example-muclp}
	\end{equation}
	The above $\mu$CLP with the query $X_1(1)$ is transformed into the following formula in the first-order fixed-point logic.
	\[ \Big(\nu X_1. \lambda x. X_1(x + 1) \land \big(\mu X_2. \lambda (x, y). y = x \lor X_2(x, y + 1)\big)(x, 0)\Big)(1) \]
\end{example}

Technically, the validity checking problem we consider is defined (in Definition~\ref{def:muCLP-with-query}) for a $\mu$CLP with a closed query, but we often focus on the $\mu$CLP itself without a query.
This is because if a query is monotone (i.e., any occurrence of predicate variables in $\psi$ is positive), then checking the validity of $(E, \psi)$ can be reduced to finding a lower bound of the solution of $E$ that satisfies $\psi$.
\begin{lemma}\label{lem:lower-bound-query}
	Let $(E, \psi)$ be a $\mu$CLP with a closed query such that any occurrence of predicate variables in $\psi$ is positive.
	Then, we have $\models (E, \psi)$ if and only if there exists a lower bound $(\mathcal{A}_1, \dots, \mathcal{A}_n)$ of the solution of $E$ such that $(\mathcal{A}_1, \dots, \mathcal{A}_n)$ satisfies $\psi$.
	\qed
\end{lemma}

If a predicate variable $X$ appears negatively in a query $\psi$, then we need a lower bound of the complement the solution for $X$ (or equivalently, an upper bound of the solution for $X$).
In this case, we consider the dual $\mu$CLP.

\begin{definition}
	Given a $\mu$CLP as in~\eqref{eq:equation-system}, we define its \emph{dual} as the list of equations obtained by replacing the $i$-th equation $X_i(\tilde{x}_i) =_{\eta_i} \phi_i$ with the following equation for each $i$ where $\lnot \mu = \nu$ and $\lnot \nu = \mu$.
	\[ X_i^{\lnot} \quad=_{\lnot \eta_i}\quad \lnot \phi_i[\lnot X_1^{\lnot}/X_1, \dots, \lnot X_n^{\lnot}/X_n] \]
\end{definition}

\begin{lemma}
	The solution of the dual $\mu$CLP is the complement of the solution of the original $\mu$CLP.
	\qed
\end{lemma}

Therefore, even if a query contains both positive and negative occurrences of predicate variables, we can reduce the validity checking problem to the case where the query contains only positive occurrences of predicate variables as follows.
First, we transform the query into the negation normal form.
Then, for each occurrence of a negated predicate variable $\lnot X_i$, we replace it with a new predicate variable $X_i^{\lnot}$ and extend the original $\mu$CLP with the dual $\mu$CLP for $X_i^{\lnot}$.
In the rest of the paper, we assume that queries contain only positive occurrences of predicate variables without loss of generality and focus on the problem of finding a lower bound of the solution of a $\mu$CLP that satisfies a given query.

\subsection{Semantics of Equation Systems}\label{sec:semantics-muCLP}

Let $\interpret{-}$ be an interpretation of sorts, function symbols, and predicate symbols.
Specifically, for each sort $s$, $\interpret{s}$ is a set of values of sort $s$.
We assume that the sort of propositions is interpreted as $\interpret{\propsort} = \{ \False, \True \}$.
For each function symbol $f$, $\interpret{f}$ is a function of type $\interpret{\sort(f)}$, and for each predicate symbol $p$, $\interpret{p}$ is a predicate of type $\interpret{\sort(p)}$ where $\interpret{-}$ is extended to lists of sorts by $\interpret{(s_1, \dots, s_k)} = \interpret{s_1} \times \cdots \times \interpret{s_k}$ and to functions by $\interpret{\tilde{s} \to s'} = \interpret{\tilde{s}} \to \interpret{s'}$.

We extend the interpretation to terms and first-order formulas in the standard way \cite{UnnoPOPL2023}.
For each term $t$ of sort $s$ with $\FreeTermVars(t) \subseteq \{ x_1, \dots, x_k \}$, the interpretation $\interpret{t}$ is defined as a function $\interpret{t} : \interpret{\sort(x_1)} \times \dots \times \interpret{\sort(x_k)} \to \interpret{s}$.
For each first-order formula $\phi$ such that $\FreePredVars(\phi) \subseteq \{ X_1, \dots, X_n \}$ and $\FreeTermVars(\phi) \subseteq \{ x_1, \dots, x_k \}$ where $\sort(X_i) = \tilde{s}_i \to \propsort$, the interpretation $\interpret{\phi}$ is defined as a function of the following type.
\[ \interpret{\phi} : \interpret{\sort(X_1)} \times \cdots \times \interpret{\sort(X_n)} \to \interpret{\sort(x_1)} \times \dots \times \interpret{\sort(x_k)} \to \{ \False, \True \} \]
We write $\mathcal{A}_1, \dots, \mathcal{A}_n; \tilde{v} \models \phi$ if $\interpret{\phi}(\mathcal{A}_1, \dots, \mathcal{A}_n)(\tilde{v}) = \True$ and $\mathcal{A}_1, \dots, \mathcal{A}_n \models \phi$ if $\mathcal{A}_1, \dots, \mathcal{A}_n; \tilde{v} \models \phi$ for every $\tilde{v}$.
We often identify $(\mathcal{A}_1, \dots, \mathcal{A}_n) \in \interpret{\sort(X_1)} \times \cdots \times \interpret{\sort(X_n)}$ with an assignment $\Theta = [X_1 \mapsto \mathcal{A}_1, \dots, X_n \mapsto \mathcal{A}_n]$ from predicate variables to their interpretations.

\begin{definition}[solution]\label{def:solution}
	Suppose we are given a $\mu$CLP as in~\eqref{eq:equation-system}.
	Note that the interpretation of the right-hand side of the $i$-th equation is given as a monotone function of the following type.
	\[ \interpret{\phi_i} : \interpret{\sort(X_1)} \times \cdots \times \interpret{\sort(X_n)} \to \interpret{\sort(X_i)} \]
	Here, we consider the order on $\{ \False, \True \}$ defined as $\False \le \True$ and extend it to the pointwise order on $\interpret{\sort(X_i)} = \interpret{\tilde{s_i}} \to \{ \False, \True \}$ for each $i$, which makes $\interpret{\sort(X_i)}$ a complete lattice.
	The solution of the $\mu$CLP is defined by solving the equations from right ($i = n$) to left ($i = 1$).
	Formally, we define $\interpret{\phi_i}^{(j)} : \interpret{\sort(X_1)} \times \cdots \times \interpret{\sort(X_j)} \to \interpret{\sort(X_i)}$ for each $i, j \in \{ 1, \dots, n \}$ by induction on $j$ from $j = n$ to $j = 0$ as follows.
	For the base case $j = n$, we define $\interpret{\phi_i}^{(n)} \coloneqq \interpret{\phi_i}$.
	For the step case, we define $\interpret{\phi_i}^{(j - 1)}$ from $\interpret{\phi_i}^{(j)}$ as follows.
	\[ \interpret{\phi_i}^{(j - 1)}(\mathcal{A}_1, \dots, \mathcal{A}_{j - 1})\ \coloneqq\ \interpret{\phi_i}^{(j)}(\mathcal{A}_1, \dots, \mathcal{A}_{j - 1}, \eta_j X_j. \interpret{\phi_j}^{(j)}(\mathcal{A}_1, \dots, \mathcal{A}_{j - 1}, X_j)) \]
	Here, for each monotone function $f : L \to L$ on a complete lattice $L$, we write $\mu X. f(X)$ (resp.\ $\nu X. f(X)$) for the least (resp.\ greatest) fixed point of $f$.
	The \emph{solution} of the $\mu$CLP is defined as $(X_1^{*}, \dots, X_n^{*}) = (\interpret{\phi_1}^{(0)}(), \dots, \interpret{\phi_n}^{(0)}())$.
	We often identify $X_i^{*} : \interpret{\tilde{s}_i} \to \{ \False, \True \}$ with the subset $\{ \tilde{x} \in \interpret{\tilde{s}_i} \mid X_i^{*}(\tilde{x}) = \True \} \subseteq \interpret{\tilde{s}_i}$.
\end{definition}

\begin{example}
	Recall the $\mu$CLP~\eqref{eq:example-muclp} in Example~\ref{ex:muCLP-to-FFL}.
	Following Definition~\ref{def:solution}, the solution is given as follows.
	Step 1: We solve the rightmost equation for $X_2$ and obtain $X_2^{*} = \{ (x, y) \mid y \geq x \}$.
	Step 2: We substitute $X_2^{*}$ for $X_2$ in the equation for $X_1$ and solve it to obtain $X_1^{*} = \{ x \mid x \ge 0 \}$.
	Note that we already have the solution $X_2^{*}$ at Step 1 in this specific example because the right-hand side of the equation for $X_2$ does not contain $X_1$.
	However, this is not the case in general: the solution of $X_2$ at Step 1 may contain $X_1$, and we need to substitute the solution of $X_1$ at Step 2 to obtain the final solution of $X_2$.
\end{example}

\subsection{Parity Graphs and Parity Games}

\begin{definition}[parity condition]
	A \emph{priority function} on a set $\mathcal{D}$ is a function $p : \mathcal{D} \to \Range{1}{n}$ that assigns a natural number (called a \emph{priority}) to each element of $\mathcal{D}$.
	A \emph{parity condition} is a pair $(p, I)$ of a priority function $p : \mathcal{D} \to \Range{1}{n}$ and a set $I \subseteq \Range{1}{n}$ of priorities.
	Let $\mathcal{D}^{\omega}$ be the set of infinite sequences over $\mathcal{D}$.
	An infinite sequence $s \in \mathcal{D}^{\omega}$ satisfies the parity condition $(p, I)$ if the minimum priority that appears infinitely often in the sequence is in $I$.
	We write the set of infinite sequences over $\mathcal{D}$ that satisfy the parity condition $(p, I)$ as $\mathbf{Parity}(p, I)$.
	\[ \mathbf{Parity}(p, I) \quad\coloneqq\quad \{ s \in \mathcal{D}^{\omega} \mid \min \mathrm{Inf}(p(s)) \in I \} \]
	Here, $p(s_0 s_1 \cdots) = p(s_0) p(s_1) \cdots$ and $\mathrm{Inf}(p_0 p_1 \cdots) \coloneqq \{ k \in \Range{1}{n} \mid k = p_i \text{ for infinitely many $i$} \}$.
\end{definition}

\begin{example}\label{ex:parity-condition-muclp}
	Given a $\mu$CLP as in~\eqref{eq:equation-system}, we often consider the parity condition $(p, I_{\nu})$ defined as follows.
	We define a set $\mathcal{D}$ as the disjoint union of the domain of each predicate variable.
	That is, $\mathcal{D} \coloneqq \bigsqcup_{i = 1}^n \mathcal{D}_i$ where $\sort(X_i) = \tilde{s_i} \to \propsort$ and $\mathcal{D}_i = \interpret{\tilde{s_i}}$ for each $i$.
	The priority function $p : \mathcal{D} \to \Range{1}{n}$ is defined by mapping each element in $\mathcal{D}_i$ to its index $i$.
	The set $I_{\nu} \subseteq \Range{1}{n}$ is defined as the set of indices of $\nu$-equations $I_{\nu} \coloneqq \{ i \mid \eta_i = \nu \}$.
\end{example}

Note that in the literature, it is common to fix $I$ to the set of even numbers $\mathbf{Even}_{\le n} = \{ 2, 4, 6, \ldots \} \subseteq \Range{1}{n}$.
We can always reduce the general case to this special case by reassigning priorities with the following lemma.
\begin{lemma}
	Let $(p, I)$ be a parity condition with $p : \mathcal{D} \to \Range{1}{n}$ and $\sigma : \Range{1}{n} \to \Range{1}{n'}$ be a monotone function.
	If $\sigma^{-1}(\sigma(I)) \subseteq I$ (i.e., $\sigma(i) \in \sigma(I)$ if and only if $i \in I$), then we have $\mathbf{Parity}(p, I) = \mathbf{Parity}(\sigma \circ p, \sigma(I))$.
	Moreover, for any parity condition $(p, I)$, there exists a monotone function $\sigma$ satisfying the above condition such that $\sigma(I)$ is the set of even priorities.
	\qed
\end{lemma}

Below, we recall the definitions of parity graphs, parity games, and parity progress measures.
The definitions mostly follow~\cite{JurdzinskiSTACS2000}, except that we allow infinite graphs and games and do not assume that every vertex has a successor, to accommodate later applications to the game semantics of $\mu$CLPs.

\begin{definition}
	A \emph{parity graph} $(V, E, p, I)$ is a (possibly infinite) directed graph $(V, E)$ equipped with a parity condition $(p, I)$ on $V$.
\end{definition}

Let $\mathbf{Ord}$ be the class of ordinals.
We write $\le_{\mathrm{lex}}$ for the (non-strict) lexicographic order on $\mathbf{Ord}^n$ for any natural number $n$.
The truncation of a tuple $(\alpha_1, \dots, \alpha_n) \in \mathbf{Ord}^n$ at $i$ is written as $(\alpha_1, \dots, \alpha_n) \downharpoonright i \coloneqq (\alpha_1, \dots, \alpha_i) \in \mathbf{Ord}^i$.
For each $i \in \Range{1}{n}$, we define a relation $\le_i$ on $\mathbf{Ord}^n$ as follows.
\[ (\alpha_1, \dots, \alpha_n) \le_i (\beta_1, \dots, \beta_n) \quad\coloneqq\quad (\alpha_1, \dots, \alpha_n) \downharpoonright i \le_{\mathrm{lex}} (\beta_1, \dots, \beta_n) \downharpoonright i \]
We also define the strict version $<_i$ of $\le_i$ in the same way.
For any $I \subseteq \Range{1}{n}$, we define a binary relation $\le_i^{I}$ on $\mathbf{Ord}^n$ by $\tilde{\alpha} \le_i \tilde{\beta}$ if $i \in I$ and $\tilde{\alpha} <_i \tilde{\beta}$ if $i \notin I$.

\begin{definition}\label{def:parity-progress-measure}
	A \emph{parity progress measure} on a parity graph $(V, E, p, I)$ with $p : V \to \Range{1}{n}$ is a function $r : V \to \mathbf{Ord}^n$ if for all $(v, w) \in E$, we have $r(v) \ge_{p(v)}^I r(w)$.
\end{definition}

\begin{proposition}\label{prop:parity-progress-measure}
	\cite[Proposition 4]{JurdzinskiSTACS2000}
	If a parity graph $(V, E, p, I)$ has a parity progress measure, then any infinite walk $v_0 v_1 \cdots \in V^{\omega}$ in the graph satisfies the parity condition $(p, I)$.
	\qed
\end{proposition}

\begin{definition}
	A \emph{parity game} $G = (V, E, p, I, V_1, V_2)$ is a parity graph $(V, E, p, I)$ equipped with a partition $V = V_1 \cup V_2$ of the set of vertices (or positions) into those of Player 1 (``Verifier'') and Player 2 (``Falsifier'').
	A \emph{play} in a parity game is a finite or infinite walk $v_0 v_1 \cdots$ in the graph.
	A play is \emph{maximal} if it is infinite or ends at a vertex with no successor.
	We write $\Play{G}$ for the set of plays in $G$ and $\MaximalPlay{G}$ for the set of maximal plays in $G$.
	We also write $\Play[v]{G}$ for the set of plays in $G$ starting from a vertex $v \in V$ and similarly for $\MaximalPlay[v]{G}$.
	A maximal play is \emph{winning} for Player 1 if (a) it is infinite and satisfies the parity condition $(p, I)$ or (b) it is finite and ends at a vertex in $V_2$ with no successor; otherwise, it is winning for Player 2.
\end{definition}

\begin{definition}\label{def:strategy}
	Let $G = (V, E, p, I, V_1, V_2)$ be a parity game.
	A \emph{(history-dependent) strategy} for Player 1 is a function $s$ defined on the set of finite plays $v_0 \cdots v_k$ ending at a vertex $v_k \in V_1$ with a successor, which assigns to each such play a successor $s(v_0 \cdots v_k)$ of $v_k$.
	If $v_k$ has no successor, then $s(v_0 \cdots v_k)$ is undefined.
	A strategy is \emph{memoryless} if $s(v_0 \cdots v_k)$ depends only on the last vertex $v_k$.
	A play $v_0 v_1 \cdots$ is \emph{consistent} with a strategy $s$ for Player 1 if for each $k$ with $v_k \in V_1$, we have $v_{k + 1} = s(v_0 \cdots v_k)$.
	We write $\StrategyPlay{\sigma}{G}$ for the set of plays in $G$ consistent with strategy $\sigma$ and $\StrategyPlay[v]{\sigma}{G}$ for the set of plays starting from a vertex $v \in V$ that are consistent with $\sigma$.
	A strategy $s$ for Player 1 is \emph{winning} from a vertex $v \in V$ if any maximal play starting from $v$ that is consistent with $s$ is winning for Player 1.
\end{definition}

\begin{definition}
	A \emph{game parity progress measure} on a parity game $G$ is a partial function $r : V \rightharpoonup \mathbf{Ord}^n$ such that for each edge $(v, w) \in E$ such that $r(v)$ is defined,
	(a) if $v \in V_1$, then there exists a successor $w$ of $v$ such that $r(w)$ is defined and $r(v) \ge_{p(v)}^I r(w)$; and
	(b) if $v \in V_2$, then for every successor $w$ of $v$, then $r(w)$ is defined and $r(v) \ge_{p(v)}^I r(w)$.
\end{definition}

\begin{proposition}
	\cite[Corollary 7]{JurdzinskiSTACS2000}
	If a parity game $G$ has a game parity progress measure $r$, then Player 1 has a winning strategy from any vertex $v$ with $r(v)$ defined.
	\qed
\end{proposition}

\subsection{Game Semantics of $\mu$CLP}

Let $E$ be a $\mu$CLP $\{ X_1(\tilde{x}_1) =_{\eta_1} \phi_1;\ \dots;\ X_n(\tilde{x}_n) =_{\eta_n} \phi_n \}$ and $\mathcal{D}_i = \interpret{\tilde{s}_i}$ for each $i$ where $\sort(X_i) = \tilde{s}_i \to \propsort$.
In what follows, we do not use the tilde notation for tuples of values in $\mathcal{D}_i$ and write $v \in \mathcal{D}_i$ instead of $\tilde{v} \in \mathcal{D}_i$ for simplicity.

The game semantics of $\mu$CLPs is given in terms of parity games called \emph{fixpoint games}~\cite{BaldanPOPL2019,BaldanCONCUR2020}.
We remark that a similar game is also considered in~\cite[Appendix D,E]{KobayashiESOP2018} for higher-order fixed-point logic.

\begin{definition}\label{def:game-semantics}
	Given a $\mu$CLP $E$, we define a parity game $\GameSemantics{E}$ played by Verifier (Player 1) and Falsifier (Player 2) as follows.
	\begin{itemize}
		\item Verifier's positions: $(X_i, v)$ where $i = 1, \dots, n$ and $v \in \mathcal{D}_i$.
		\item Falsifier's positions: a tuple $(\mathcal{A}_1, \dots, \mathcal{A}_n)$ where $\mathcal{A}_j \subseteq \mathcal{D}_j$ for each $j$.
		\item Verifier's moves: $(X_i, v) \to (\mathcal{A}_1, \dots, \mathcal{A}_n)$ such that $\interpret{\phi_i}(\mathcal{A}_1, \dots, \mathcal{A}_n)(v)$ is true.
		\item Falsifier's moves: $(\mathcal{A}_1, \dots, \mathcal{A}_n) \to (X_j, v)$ such that $v \in \mathcal{A}_j$.
		\item Parity condition: $p(X_i, v) = i$, $p(\mathcal{A}_1, \dots, \mathcal{A}_n) = n + 1$, and $I = \{ i \mid \eta_i = \nu \}$.
	\end{itemize}
\end{definition}

The intuition of the game is as follows.
At a Verifier's position $(X_i, v)$, Verifier tries to show $v \in X^{*}_i$ where $(X^{*}_1, \dots, X^{*}_n)$ is the solution of $E$.
To do so, Verifier unfolds the fixed point and chooses an assignment $(\mathcal{A}_1, \dots, \mathcal{A}_n)$ for predicate variables that satisfies the right-hand side formula $\phi_i$ at $v$.
Then, Verifier claims ``$(\mathcal{A}_1, \dots, \mathcal{A}_n)$ is such that for any $j$ and $u \in \mathcal{A}_j$, the solution of $E$ satisfies $u \in X^{*}_j$''.
On the other hand, Falsifier chooses one of the sets $\mathcal{A}_j$ and an element $v \in \mathcal{A}_j$ to challenge Verifier to show $v \in X^{*}_j$.
The parity condition for Verifier's positions is defined as in Example~\ref{ex:parity-condition-muclp}.
The priority of Falsifier's positions is set to $p(\mathcal{A}_1, \dots, \mathcal{A}_n) = n + 1$ so that the parity condition depends only on the priorities of Verifier's positions.
This ensures that when Verifier keeps responding to Falsifier's challenges forever, if $\mu$-equations are unfolded infinitely often, then Verifier loses the game, and if $\nu$-equations are unfolded infinitely often, then Verifier wins the game.
The game semantics coincides with the semantics defined in Section~\ref{sec:semantics-muCLP}.

\begin{theorem}\label{thm:game-semantics}
	\cite[Theorem~4.8]{BaldanPOPL2019}
	Let $(X_1^{*}, \dots, X_n^{*})$ be the solution of $E$.
	Then, $X_i^{*}(v)$ is true if and only if Verifier wins from $(X_i, v)$ in $\GameSemantics{E}$.
	\qed
\end{theorem}

%% file: ppm-translation.tex
\section{Finding Lower Bounds for $\mu$CLPs}\label{sec:validity-checking}

Throughout this section and Sections~\ref{sec:comparison-popl23} and~\ref{sec:comparison-popl25}, let $E$ be a $\mu$CLP $\{ X_1(x_1) =_{\eta_1} \phi_1;\ \dots;\ X_n(x_n) =_{\eta_n} \phi_n \}$ and $\mathcal{D}_i = \interpret{\tilde{s}_i}$ for each $i$ where $\sort(X_i) = \tilde{s}_i \to \propsort$.
Recall from Example~\ref{ex:parity-condition-muclp} that $E$ induces a parity condition on the disjoint union $\mathcal{D} \coloneqq \bigsqcup_{i=1}^n \mathcal{D}_i$.
Unless stated otherwise, we use this parity condition throughout these sections.

\subsection{Parity Relations}

We introduce the notion of parity relations, which is a generalization of well-founded relations and ensures the parity condition on infinite sequences of values in $\mathcal{D}$.

\begin{definition}\label{def:parity-relation}
	Let $R \subseteq \mathcal{D} \times \mathcal{D}$ be a relation.
	We say that $R$ is an \emph{$(p, I)$-parity progress relation} (or \emph{$(p, I)$-parity relation} for short) if any infinite sequence $x_0 x_1 \cdots \in \mathcal{D}^{\omega}$ such that $R(x_m, x_{m + 1})$ for each $m$ satisfies the parity condition $(p, I)$.
\end{definition}

\begin{example}\label{ex:well-founded-relation-as-parity-relation}
	Well-founded relations are equivalent to $(p, \emptyset)$-parity relations where $p : \mathcal{D} \to \{ 1 \}$ is the trivial priority function.
	In this case, for any partial function $r : \mathcal{D} \rightharpoonup \mathbf{Ord}$, the following relation is a $(p, \emptyset)$-parity relation (i.e.\ well-founded relation).
	Here, we write $r(x) {\downarrow}$ if $r(x)$ is defined.
	\[ R(x, y) \quad\coloneqq\quad r(x) {\downarrow} \land r(y) {\downarrow} \land r(x) > r(y) \]
\end{example}

The construction of well-founded relations in Example~\ref{ex:well-founded-relation-as-parity-relation} can be generalized to construct parity relations by replacing ranking functions with parity progress measures.
\begin{example}\label{ex:on-exit-progress-measure}
	Let $p : \mathcal{D} \to \{ 1, \dots, n \}$ be a priority function and $I \subseteq \{ 1, \dots, n \}$.
	For any partial function $r : \mathcal{D} \rightharpoonup \mathbf{Ord}^n$, the relation $R_{i, j}$ defined below is a $(p, I)$-parity relation.
	\begin{equation}
		R(x, y) \quad\coloneqq\quad \begin{cases}
			r(x) {\downarrow} \land r(y) {\downarrow} \land r(x) >_{p(x)} r(y) & \text{if $p(x) \notin I$} \\
			r(x) {\downarrow} \land r(y) {\downarrow} \land r(x) \ge_{p(x)} r(y) & \text{if $p(x) \in I$}
		\end{cases}
		\label{eq:on-exit-progress-measure-relation}
	\end{equation}
	The following relation, which uses $p(y)$ instead of $p(x)$, is also a $(p, I)$-parity relation.
	\begin{equation}
		R(x, y) \quad\coloneqq\quad \begin{cases}
			r(x) {\downarrow} \land r(y) {\downarrow} \land r(x) >_{p(y)} r(y) & \text{if $p(y) \notin I$} \\
			r(x) {\downarrow} \land r(y) {\downarrow} \land r(x) \ge_{p(y)} r(y) & \text{if $p(y) \in I$}
		\end{cases}
		\label{eq:on-entry-progress-measure-relation}
	\end{equation}
\end{example}

When $\mathcal{D} \coloneqq \bigsqcup_{i = 1}^n \mathcal{D}_i$ is a disjoint union of the sets $\mathcal{D}_i = p^{-1}(i)$ as in Example~\ref{ex:parity-condition-muclp}, we often represent a relation $R \subseteq \mathcal{D} \times \mathcal{D}$ as a family of relations $\{ R_{i, j} \subseteq \mathcal{D}_i \times \mathcal{D}_j \}_{i, j = 1, \dots, n}$ where $R_{i, j}$ is the restriction of $R$ to $\mathcal{D}_i \times \mathcal{D}_j$.
In this case, \eqref{eq:on-exit-progress-measure-relation} can be defined from a family of partial functions $\{ r_i : \mathcal{D}_i \rightharpoonup \mathbf{Ord}^n \}_{i = 1, \dots, n}$ as follows.
\begin{equation}
	R_{i, j}(x, y) \quad\coloneqq\quad \begin{cases}
		r_i(x) {\downarrow} \land r_j(y) {\downarrow} \land r_i(x) >_{i} r_j(y) & \text{if $i \notin I$} \\
		r_i(x) {\downarrow} \land r_j(y) {\downarrow} \land r_i(x) \ge_{i} r_j(y) & \text{if $i \in I$}
	\end{cases}
	\label{eq:on-exit-progress-measure-relation-disjoint-union}
\end{equation}

We will discuss a more sophisticated way to provide parity relations in Section~\ref{sec:implementation}, where we analyse the dependency between equations and reassign priorities to variables so that we obtain a compact representation of parity relations.

\subsection{A Novel Transformation Based on Parity Relations}\label{sec:main-theorems}

\begin{definition}
	We say an assignment $(\mathcal{A}_1, \dots, \mathcal{A}_n)$ for predicate variables $(X_1, \dots, X_n)$ is a \emph{post-fixed point} of a $\mu$CLP $E$ if it satisfies the following condition.
	\begin{equation}
		\mathcal{A}_1, \dots, \mathcal{A}_n;\ v \quad\models\quad X_i(x_i) \implies \phi_i \qquad\quad \text{for each $i = 1, \dots, n$ and $v \in \mathcal{D}_i$}
	\end{equation}
\end{definition}

Let $E_{\nu}$ be the $\mu$CLP obtained by replacing each $\mu$ in $E$ with $\nu$.
By the Knaster--Tarski theorem, any post-fixed point of $E$ is a lower bound of the solution of $E_{\nu}$.
Thus, if $E$ is a $\mu$CLP such that $\eta_i = \nu$ for all $i$, then the problem of finding a lower bound of the solution of $E$ can be reduced to the problem of finding a post-fixed point of $E$, which can be expressed as constraint solving over predicate variables.
However, this does not directly apply if $E$ contains $\mu$-equations.
To handle $\mu$-equations, we consider restricting the occurrences of predicate variables in the right-hand sides using parity relations.

\begin{definition}\label{def:parity-relation-transformation}
	Let $R = \{ R_{i, j} \}_{i, j = 1, \dots, n}$ be a parity relation with respect to the parity condition in Example~\ref{ex:parity-condition-muclp}.
	We define a $\mu$CLP $E[R]$ by replacing each occurrence of $X_j(t)$ in $\phi_i$ with $X_j(t) \land R_{i, j}(x_i, t)$.
	That is, the $i$-th equation $X_i(x_i) =_{\eta_i} \phi_i$ in $E$ is replaced with the following equation for each $i$.
	\[ X_i(x_i) \quad=_{\nu}\quad \phi_i[\lambda y_1. X_1(y_1) \land R_{i, 1}(x_i, y_1)/X_1, \dots, \lambda y_n. X_n(y_n) \land R_{i, n}(x_i, y_n)/X_n] \]
	Here, we implicitly introduce a fresh predicate symbol $R_{i, j}$ for each $i, j$ whose interpretation is given by the parity relation $R_{i, j}$.
\end{definition}

The $\mu$-to-$\nu$ transformation $E \mapsto E[R]$ gives us a desired reduction for general $\mu$CLPs as stated in the following theorems.
\begin{theorem}[soundness]\label{thm:soundness}
	If $R$ is a parity relation, then any post-fixed point of $E[R]$ is a lower bound of the solution of $E$.
\end{theorem}
Theorem~\ref{thm:soundness} can be proved by showing the existence of a memoryless winning strategy for Verifier.
Given a parity relation $R$, consider a sub-game of $G_E$, in which Verifier's moves are restricted to the following:
\[ (X_i, v) \quad\to\quad (\mathcal{A}_1, \dots, \mathcal{A}_n) \quad\text{where}\quad \forall j. \mathcal{A}_j \subseteq R_{i, j}(v) \]
Here, we write $R_{i, j}(v) = \{ u \mid R_{i, j}(v, u) \}$.
In this game, any infinite play satisfies the parity condition since $R$ is a parity relation.
Thus, Verifier is winning unless Verifier gets stuck.
Because only Verifier's moves are restricted here, if Verifier wins in the sub-game, then Verifier also wins in the original game $G_E$.
Moreover, the existence of a post-fixed point of $E[R]$ implies that Verifier never gets stuck in the sub-game, as we see below.

\begin{proof}[Proof of Theorem~\ref{thm:soundness}]
	If $(\mathcal{A}_1, \dots, \mathcal{A}_n)$ is a post-fixed point of $E[R]$:
	\[ v \in \mathcal{A}_i \quad\implies\quad \interpret{\phi_i}(\mathcal{A}_1 \cap R_{i, 1}(v), \dots, \mathcal{A}_n \cap R_{i, n}(v))(v) \text{ is true} \quad\text{for each $i$} \]
	and if the current position $(X_i, v)$ is such that $v \in \mathcal{A}_i$,
	then Verifier can choose the next move as $(X_i, v) \to (\mathcal{A}_1 \cap R_{i, 1}(v), \dots, \mathcal{A}_n \cap R_{i, n}(v))$.
	No matter how Falsifier responds, the next position is again of the form $(X_j, u)$ such that $u \in \mathcal{A}_j$, and the play continues indefinitely unless Falsifier gets stuck.
	Since $R$ is a parity relation, any infinite play consistent with this strategy satisfies the parity condition.
	Thus, Verifier is winning from $(X_i, v)$ such that $v \in \mathcal{A}_i$.
	By Theorem~\ref{thm:game-semantics}, we conclude $v \in X_i^{*}$.
\end{proof}

\iflong
\begin{remark}
	There is an alternative proof of Theorem~\ref{thm:soundness} without game semantics.
	Let $E'[R]$ be the $\mu$CLP obtained in the same way as $E[R]$ except that we keep $\mu$ and $\nu$ in the original equations.
	Then, we can show that the solution of $E'[R]$ is equal to the solution of $E[R]$.
	By definition, the solution of $E'[R]$ is a lower bound of the solution of $E$.
	Therefore, any post-fixed point of $E[R]$ is a lower bound of the solution of $E$.
\end{remark}
\fi

The $\mu$-to-$\nu$ transformation is also complete in the following sense.

\begin{theorem}[completeness]\label{thm:completeness}
	There exists a parity relation $R$ such that the solution of $E$ is a post-fixed point of $E[R]$.
\end{theorem}
\begin{proof}
	See
	\iflong
	Appendix~\ref{sec:complete-progress-measure}
	\else
	Appendix A
	\fi
	for the construction of such a parity relation $R$.
\end{proof}

\begin{example}
	Consider the $\mu$CLP $E$ in Example~\ref{ex:muCLP-to-FFL} again.
	By Theorem~\ref{thm:soundness} and \ref{thm:completeness} (and Lemma~\ref{lem:lower-bound-query}), the validity of $E$ with a query $X_1(1)$ can be reduced to the problem of finding an assignment to $X_1$ and $X_2$ such that $X_1(1)$ holds and $(X_1, X_2)$ is a post-fixed point of $E[R]$ for some parity relation $R$:
	\begin{align}
		X_1(x) \quad&\implies\quad X_1(x + 1) \land R_{1, 1}(x, x + 1) \quad\land\quad X_2(x, 0) \land R_{1, 2}(x, (x, 0)) \\
		X_2(x, y) \quad&\implies\quad x = y \quad\lor\quad \big(X_2(x, y + 1) \land R_{2, 2}((x, y), (x, y + 1))\big)
	\end{align}
	The solution for the latter problem is given as follows.
	Let $R = \{ R_{i, j} \}_{i, j = 1, \dots, 2}$ be a parity relation of the form~\eqref{eq:on-exit-progress-measure-relation-disjoint-union} where a pair of partial functions is given by $r_1(x) = 0$ for any $x$ and $r_2(x, y) = x - y$ if $x \ge y$.
	Concretely, $R$ is given as follows.
	\begin{gather}
		R_{1, 1}(x, x') \quad=\quad \True, \qquad\qquad
		R_{1, 2}(x, (x', y')) \quad=\quad x' \ge y', \\
		R_{2, 2}((x, y), (x', y')) \quad=\quad x \ge y \land x' \ge y' \land x - y > x' - y'
	\end{gather}
	Then, it is easy to see that an assignment $[X_1(x) \mapsto x \ge 0,\ X_2(x, y) \mapsto x \ge y]$ is a post-fixed point of $E[R]$ and hence a lower bound of the solution of $E$.
\end{example}

%% file: popl23.tex
\section{Game Interpretation of the Transformation in \cite{UnnoPOPL2023}}
\label{sec:comparison-popl23}

We rephrase the winning condition for Verifier in the game semantics of $\mu$CLP in terms of Streett conditions and well-founded relations, and define a new two-player game for $\mu$CLP based on this characterization.
Then, we show that the transformation proposed in \cite{UnnoPOPL2023} can be interpreted as a technique to find a winning strategy in this game.

\subsection{Streett Relations}

\begin{definition}[Streett condition]
	A \emph{Streett condition} is specified by a finite family of pairs $\{ (A_i, B_i) \}_{i = 1}^m$ of subsets $A_i, B_i \subseteq \mathcal{D}$.
	An infinite sequence $x_0 x_1 \cdots$ in $\mathcal{D}^{\omega}$ satisfies the Streett condition if for all $i = 1, \dots, m$, if $x_j \in A_i$ for infinitely many $j$'s, then $x_j \in B_i$ for infinitely many $j$'s.
\end{definition}

It is well-known that the parity condition can be expressed as a Streett condition.

\begin{proposition}\label{prop:parity-condition-streett-condition}
	Let $p : \mathcal{D} \to \{ 1, \dots, n \}$ be a priority function and $I \subseteq \{ 1, \dots, n \}$.
	Then, the parity condition for $(p, I)$ is equivalent to the Streett condition $\{ (\mathcal{D}_{\le i}, \mathcal{D}_{< i}) \mid i \notin I \}$ where $\mathcal{D}_{\le k} \coloneqq \{ x \in \mathcal{D} \mid p(x) \le k \}$ and $\mathcal{D}_{< k} \coloneqq \{ x \in \mathcal{D} \mid p(x) < k \}$.
	\qed
\end{proposition}

Similarly to parity relations (Definition~\ref{def:parity-relation}), we define the notion of Streett relations, which ensure the Streett condition on infinite sequences of values in $\mathcal{D}$.
\begin{definition}[Streett relation]
	Let $A, B \subseteq \mathcal{D}$ and $R \subseteq \mathcal{D} \times \mathcal{D}$.
	The relation $R$ is \emph{$(A, B)$-Streett} if for any infinite sequence $x_0 x_1 \cdots$ in $\mathcal{D}^{\omega}$ such that $(x_m, x_{m+1}) \in R$ for every $m \ge 0$, the sequence satisfies the Streett condition $\{ (A, B) \}$.
	The relation $R$ is \emph{Streett} for a family of pairs $\{ (A_i, B_i) \}_{i = 1}^m$ if $R$ is $(A_i, B_i)$-Streett for every $i = 1, \dots, m$.
\end{definition}

As an immediate consequence of Proposition~\ref{prop:parity-condition-streett-condition}, a relation $R$ is $(p, I)$-parity if and only if $R$ is Streett for the family of pairs $\{ (\mathcal{D}_{\le i}, \mathcal{D}_{< i}) \mid i \notin I \}$.

\begin{example}
	For any partial function $r : \mathcal{D} \rightharpoonup \mathbf{Ord}$, the following relation $R$ is $(A, B)$-Streett.
	\[ R(x, y) \quad\coloneqq\quad \begin{cases}
		r(x) {\downarrow} \land r(y) {\downarrow} \land r(x) > r(y) & \text{if $x \in A \setminus B$} \\
		r(x) {\downarrow} \land r(y) {\downarrow} & \text{if $x \in B$} \\
		r(x) {\downarrow} \land r(y) {\downarrow} \land r(x) \ge r(y) & \text{if $x \notin A \cup B$}
	\end{cases} \]
	When $A = \mathcal{D} \setminus B$, then the Streett condition for $(A, B)$ is equivalent to the B\"uchi condition for $B$.
	Thus, the function $r$ above can be seen as a generalization of B\"uchi ranking functions \cite{ChatterjeeFM2025} to the case of Streett conditions.
\end{example}

Well-founded relations are a special case of Streett relations where $A = \mathcal{D}$ and $B = \emptyset$.
In general, a Streett relation $R$ can be seen as a generalization of well-founded relations in the following sense: when we restrict $R$ to $\mathcal{D} \setminus B$, then $R$ guarantees that there is no sequence $x_0 \mathrel{R} x_1 \mathrel{R} \cdots$ that contain infinitely many elements in $A \setminus B$.
This observation leads to the following characterization of Streett relations in terms of well-founded relations.

\begin{lemma}\label{lem:streett-relation-well-founded-relation}
	A relation $R$ is $(A, B)$-Streett if and only if the relation $\hat{R}_{A, B}$ defined below is well-founded:
	\[ \hat{R}_{A, B} \coloneqq \{ (x_0, x_m) \mid m > 0 \land x_0 \mathrel{R} x_1 \mathrel{R} \cdots \mathrel{R} x_{m-1} \mathrel{R} x_m \land x_0, x_m \in A \setminus B \land x_1, \dots, x_{m - 1} \notin A \cup B \}
	\iflong\else
	\tag*{\qed}
	\fi \]
\end{lemma}
\iflong
\begin{proof}
	For notational simplicity, we write $\hat{R}$ for $\hat{R}_{A, B}$.
	If $\hat{R}$ is not well-founded, then there exists an infinite sequence $x_0 x_1 \ldots$ such that $(x_i, x_{i+1}) \in \hat{R}$ for every $i \ge 0$.
	By definition of $\hat{R}$, this sequence can be expanded to an infinite sequence $y_0 y_1 \ldots$ such that $(y_i, y_{i+1}) \in R$ for every $i \ge 0$, $y_i \in A$ for infinitely many $i$'s, but $y_i \notin B$ for any $i$.
	Hence, $R$ is not $(A, B)$-Streett.

	Conversely, if $R$ is not $(A, B)$-Streett, then there exists an infinite sequence $x_0 x_1 \ldots$ such that $(x_i, x_{i+1}) \in R$ for every $i \ge 0$, $x_i \in A$ for infinitely many $i$'s, but $x_i \in B$ for only finitely many $i$'s.
	Without loss of generality, we can assume that $x_i \notin B$ for every $i \ge 0$.
	Then, there exists an infinite subsequence $x_{i_0} x_{i_1} \ldots$ of $x_0 x_1 \ldots$ such that $x_{i_j} \in A$ for every $j \ge 0$ and $(x_{i_j}, x_{i_{j + 1}}) \in \hat{R}$ for every $j \ge 0$.
	Thus, $\hat{R}$ is not well-founded.
\end{proof}
\fi

\subsection{Winning Condition in Terms of Well-Founded Relations}
\label{sec:winning-condition-well-founded-relations}
Observe that the winning condition for Verifier in the game semantics of $\mu$CLP (Definition~\ref{def:game-semantics}) can be characterized as follows.

\begin{lemma}\label{lem:winning-condition-parity-relation}
	Let $E$ be a $\mu$CLP and $\GameSemantics{E}$ be the game semantics of $E$ (Definition~\ref{def:game-semantics}).
	For each strategy $\sigma$ for Verifier and a Verifier's position $(X_i, v)$, we define a relation $R^{\sigma}_{(X_i, v)} \subseteq \mathcal{D} \times \mathcal{D}$ as follows:
	\begin{align}
		R^{\sigma}_{(X_i, v)} \quad\coloneqq\quad \{ (u, w) &\mid (X_i, v) \to \cdots \to (X_j, u) \to (\mathcal{A}_1, \dots, \mathcal{A}_n) \to (X_k, w) \in \StrategyPlay[(X_i, v)]{\sigma}{\GameSemantics{E}} \}
	\end{align}
	Verifier is winning from $(X_i, v)$ if and only if there exists a memoryless strategy $\sigma$ such that $R^{\sigma}_{(X_i, v)}$ is a parity relation and any finite maximal play consistent with $\sigma$ ends at a Falsifier's position.
\end{lemma}
\begin{proof}
	The if part is obvious.
	The only if part follows from the follows from completeness (Theorem~\ref{thm:completeness}) and the proof of soundness (Theorem~\ref{thm:soundness}) of the transformation in Section~\ref{sec:main-theorems}.
\end{proof}

By Proposition~\ref{prop:parity-condition-streett-condition}, the parity relation $R^{\sigma}_{(X_i, v)}$ can be rephrased as a Streett relation, and by Lemma~\ref{lem:streett-relation-well-founded-relation}, it is further rephrased in terms of well-founded relations.
Then, we can rephrase the winning condition for Verifier as follows:
Verifier is winning from $(X_i, v)$ in the game semantics of $\mu$CLP $E$ if and only if there exists a memoryless strategy $\sigma$ such that any finite maximal play consistent with $\sigma$ ends at a Falsifier's position and for each $\eta_j = \mu$, $(\widehat{R^{\sigma}_{(X_i, v)}})_{\mathcal{D}_{\le j}, \mathcal{D}_{< j}}$ is well-founded.
By definition, the well-foundedness condition for $(\widehat{R^{\sigma}_{(X_i, v)}})_{\mathcal{D}_{\le j}, \mathcal{D}_{< j}}$ is stated as follows:
there exists a well-founded relation $R$ such that for any play $\rho_1, \dots, \rho_m$ consistent with $\sigma$, if they are of the form
\begin{align}
	\rho_k \quad=\quad (X_i, v) \to \cdots \to (X_{j_{k - 1}}, u_{k - 1}) \to (\mathcal{A}_{k, 1}, \dots, \mathcal{A}_{k, n}) \to (X_{j_{k}}, u_{k})
\end{align}
with $j_0 = j_m = j$ and $j_1, \dots, j_{m - 1} > j$, then $(u_0, u_m) \in R$.
We can combine the plays $\rho_1, \dots, \rho_m$ into a single play.
\begin{align}
	&(X_i, v) \to && \cdots &&\to (X_{j_0}, u_0) && \qquad j_0 = j \\
	&&&\to (\mathcal{A}_{1, 1}, \dots, \mathcal{A}_{1, n}) &&\to (X_{j_1}, u_1) && \qquad j_1 > j \\
	&&&\qquad\vdots \\
	&&&\to (\mathcal{A}_{m-1, 1}, \dots, \mathcal{A}_{m-1, n}) &&\to (X_{j_{m-1}}, u_{m-1}) && \qquad j_{m-1} > j \\
	&&&\to (\mathcal{A}_{m, 1}, \dots, \mathcal{A}_{m, n}) &&\to (X_{j_m}, u_m) && \qquad j_m = j \label{eq:play-transitive-closure}
\end{align}
This play is also consistent with $\sigma$ because $\sigma$ is memoryless.
This observation leads to the following characterization of the winning condition.
\begin{corollary}\label{cor:winning-condition-well-founded-relation}
	Verifier is winning from $(X_i, v)$ in the game semantics (Definition~\ref{def:game-semantics}) if and only if there exist a memoryless strategy $\sigma$ and a family of well-founded relations $R = \{ R_j \subseteq \mathcal{D}_j \times \mathcal{D}_j \}_{j \in \{ j \mid \eta_j = \mu \}} $ such that (a) for any play of the form~\eqref{eq:play-transitive-closure}, if it is consistent with $\sigma$, then $(u_0, u_m) \in R_j$ and (b) any finite maximal play consistent with $\sigma$ ends at a Falsifier's position.
	\qed
\end{corollary}

Based on the above characterization, we define the following two-player game for $\mu$CLP.
\begin{definition}\label{def:well-founded-game}
	Given a $\mu$CLP $E = \{ X_1(x_1) =_{\eta_1} \phi_1; \dots; X_n(x_n) =_{\eta_n} \phi_n \}$, we define a game $\WellFoundedGame{E}$ played by Verifier (Player 1) and Falsifier (Player 2) as follows.
	\begin{itemize}
		\item Verifier's positions: $(X_i, v; h_1, \dots, h_{i - 1})$ where $v \in \mathcal{D}_i$ and $h_j \in \mathcal{D}_j \cup \{ \bot \}$ for each $j$.
		\item Falsifier's positions: $(\mathcal{A}_1, \dots, \mathcal{A}_n; h_1, \dots, h_n)$ where $\mathcal{A}_j \subseteq \mathcal{D}_j$ and $h_j \in \mathcal{D}_j \cup \{ \bot \}$ for each $j$.
		\item Verifier's moves: $(X_i, v; h_1, \dots, h_{i - 1}) \to (\mathcal{A}_1, \dots, \mathcal{A}_n; h_1, \dots, h_{i - 1}, v, \bot, \dots, \bot)$ such that $\interpret{\phi_i}(\mathcal{A}_1, \dots, \mathcal{A}_n)(v)$ is true.
		\item Falsifier's moves: $(\mathcal{A}_1, \dots, \mathcal{A}_n; h_1, \dots, h_n) \to (X_j, v; h_1, \dots, h_{j - 1})$ such that $v \in \mathcal{A}_j$.
	\end{itemize}
\end{definition}
The game in Definition~\ref{def:well-founded-game} is similar to Definition~\ref{def:game-semantics} except that the components $h_1, \dots, h_n$ are added to each position.
For each $i$, the component $h_i$ is used to remember the last visit to $(X_i, v)$, but it is reset to $\bot$ when $(X_j, u)$ for some $j < i$ is visited after the last visit to $(X_i, v)$.
Thus, the components $h_1, \dots, h_n$ maintain the partial history information required in Corollary~\ref{cor:winning-condition-well-founded-relation}.
Although only the components $h_j$ with $\eta_j = \mu$ will be used in the winning condition, we also keep the unnecessary components corresponding to $\eta_j = \nu$ for simplicity and readability.

\begin{definition}[strategy and winning condition]
	We define a \emph{strategy} for $\WellFoundedGame{E}$ (Definition~\ref{def:well-founded-game}) in the same way as that for parity games (Definition~\ref{def:strategy}).
	A strategy is \emph{strongly memoryless} if it is memoryless and does not depend on $h_1, \dots, h_n$.
	We say that \emph{Verifier wins from a position $(X_i, v; h_1, \dots, h_i)$} if there exist a family of well-founded relations $R = \{ R_j \}_{j \in \{ j \mid \eta_j = \mu \}}$ and a strategy for Verifier such that for each maximal play $\rho$ consistent with the strategy,
	\begin{itemize}
		\item if $\rho$ is finite, then it ends in a Falsifier's position (i.e., Verifier never gets stuck), and
		\item if $\rho$ is infinite and $\eta_j = \mu$, then for each Falsifier's move
		\[ (\mathcal{A}_1, \dots, \mathcal{A}_n; h_1, \dots, h_n) \to (X_j, u; h_1, \dots, h_{j - 1}) \qquad \text{in $\rho$} \]
		we have $(h_j, u) \in R_j$ if $h_j \neq \bot$.
	\end{itemize}
	Note that the family of well-founded relations $R$ is fixed across possible plays consistent with the strategy.
	Therefore, the winning condition for Verifier is defined in terms of all plays consistent with the strategy, rather than individual plays.
	We also do not define the winning condition for Falsifier.
\end{definition}

In what follows, we fix the initial position in $\WellFoundedGame{E}$ as $(X_i, v; \bot, \dots, \bot)$.
By a slight abuse of notation, we write the set of (finite or infinite) plays in $\WellFoundedGame{E}$ starting from a position $(X_i, v; h_1, \dots, h_n)$ as $\Play[(X_i, v; h_1, \dots, h_n)]{\WellFoundedGame{E}}$.
It is easy to see that if $\rho \in \Play[(X_i, v; \bot, \dots, \bot)]{\WellFoundedGame{E}}$, then by removing $h_1, \dots, h_n$ from each position in $\rho$, we can obtain a play $\rho' \in \Play[(X_i, v)]{\GameSemantics{E}}$ where $\GameSemantics{E}$ is the game semantics defined in Definition~\ref{def:game-semantics}.
In fact, this mapping is a bijection between $\Play[(X_i, v; \bot, \dots, \bot)]{\WellFoundedGame{E}}$ and $\Play[(X_i, v)]{\GameSemantics{E}}$ since we can uniquely recover $h_1, \dots, h_n$ from $\rho' \in \Play[(X_i, v)]{\GameSemantics{E}}$.
Moreover, this bijection induces a bijection between (history-dependent) strategies for $\WellFoundedGame{E}$ and those for $\GameSemantics{E}$, assuming that the initial positions are fixed as $(X_i, v; \bot, \dots, \bot)$ and $(X_i, v)$, respectively.
Thus, we often identify a strategy for $\WellFoundedGame{E}$ with the corresponding strategy for $\GameSemantics{E}$.
In particular, memoryless strategies for $\GameSemantics{E}$ correspond to strongly memoryless strategies for $\WellFoundedGame{E}$.

Moreover, if a play in $\Play[(X_i, v)]{\GameSemantics{E}}$ has the form~\eqref{eq:play-transitive-closure}, then the last move of the corresponding play in $\Play[(X_i, v; \bot, \dots, \bot)]{\WellFoundedGame{E}}$ has the form
$(\mathcal{A}_{m, 1}, \dots, \mathcal{A}_{m, n}; h_1, \dots, h_n) \to (X_{j_m}, u_m; h_1, \dots, h_{j_m - 1})$
where $h_{j_m} = u_0$ and $j_m = j$.
Conversely, if a play in $\Play[(X_i, v; \bot, \dots, \bot)]{\WellFoundedGame{E}}$ has the last move of the form
$(\mathcal{A}_1, \dots, \mathcal{A}_n; h_1, \dots, h_n) \to (X_j, u; h_1, \dots, h_{j - 1})$
with $h_j \neq \bot$, then the corresponding play in $\Play[(X_i, v)]{\GameSemantics{E}}$ has the form~\eqref{eq:play-transitive-closure} with $u_0 = h_j$ and $u_m = u$.
By Corollary~\ref{cor:winning-condition-well-founded-relation}, this correspondence gives the following proposition.

\begin{proposition}\label{prop:strongly-memoryless-well-founded-game-parity-game}
	Verifier wins from a position $(X_i, v)$ in $\GameSemantics{E}$ (Definition~\ref{def:game-semantics}) if and only if Verifier wins from a position $(X_i, v; \bot, \dots, \bot)$ in $\WellFoundedGame{E}$ with a strongly memoryless strategy.
	\iflong\else
	\qed
	\fi
\end{proposition}
\iflong
\begin{proof}
	If Verifier wins from a position $(X_i, v)$ in $\GameSemantics{E}$, then by Corollary~\ref{cor:winning-condition-well-founded-relation}, there exists a memoryless strategy $\sigma$ and a well-founded relation $R$ that satisfy the condition in Corollary~\ref{cor:winning-condition-well-founded-relation}.
	We show that the corresponding strongly memoryless strategy for Verifier in $\WellFoundedGame{E}$ witnesses that Verifier wins from $(X_i, v; \bot, \dots, \bot)$.
	For any maximal play $\rho$ consistent with the strategy $\sigma$, if $\rho$ is finite, then it ends in a Falsifier's position.
	If $\rho$ is infinite, then for each Falsifier's move
	\[ (\mathcal{A}_1, \dots, \mathcal{A}_n; h_1, \dots, h_n) \to (X_j, u; h_1, \dots, h_{j - 1}) \]
	in $\rho$, if $\eta_j = \mu$ and $h_j \neq \bot$, then it must be the case that the corresponding play in $\GameSemantics{E}$ has the prefix of the form~\eqref{eq:play-transitive-closure} with $u_0 = h_j$ and $u_m = u$, and hence $(h_j, u) \in R_j$.

	Conversely, suppose that Verifier wins from a position $(X_i, v; \bot, \dots, \bot)$ with a strongly memoryless strategy in $\WellFoundedGame{E}$.
	We show that the corresponding memoryless strategy for Verifier in $\GameSemantics{E}$ witnesses that Verifier wins from $(X_i, v)$.
	For any finite maximal play $\rho$ in $\GameSemantics{E}$ consistent with the strategy, it ends in a Falsifier's position.
	For any play $\rho$ in $\GameSemantics{E}$ consistent with the strategy, if it has the form~\eqref{eq:play-transitive-closure}, then the last move of the corresponding play in $\WellFoundedGame{E}$ is of the form
	\[ (\mathcal{A}_{m, 1}, \dots, \mathcal{A}_{m, n}; h_1, \dots, h_n) \to (X_{j_m}, u_m; h_1, \dots, h_{j_m - 1}) \]
	with $h_{j_m} = u_0$ and $j_m = j$, and hence $(h_j, u_m) \in R_j$ if $\eta_j = \mu$.
	By Corollary~\ref{cor:winning-condition-well-founded-relation}, this implies that Verifier wins from $(X_i, v)$ in $\GameSemantics{E}$.
\end{proof}
\fi

The proof of Proposition~\ref{prop:strongly-memoryless-well-founded-game-parity-game} (specifically, Corollary~\ref{cor:winning-condition-well-founded-relation}) depends on strong memorylessness of the strategy for $\WellFoundedGame{E}$.
However, the result can be generalized to history-dependent strategies.

In Lemma~\ref{lem:streett-relation-well-founded-relation}, we only consider infinite sequences $x_0 x_1 \dots$ such that $(x_m, x_{m+1}) \in R$ for every $m$ where $R$ is a binary relation.
For an arbitrary infinite sequence, we have the following lemma.
\begin{lemma}\label{lem:sequence-streett-condition-well-founded-relation}
	Let $R$ be a well-founded relation on $\mathcal{D}$ and $x_0 x_1 \cdots \in \mathcal{D}^{\omega}$ be an infinite sequence.
	Assume that for any $l$ and $m$ with $l < m$, if $x_l, x_m \in A \setminus B$ and $x_{l + 1}, \dots, x_{m - 1} \notin A \cup B$, then $(x_l, x_m) \in R$.
	Then, the sequence $x_0 x_1 \cdots$ satisfies the Streett condition $\{ (A, B) \}$.
\end{lemma}
\begin{proof}
	Suppose that $x_0 x_1 \dots \in T$ violates the Streett condition.
	Then, there exists an infinite sequence $m_0 < m_1 < m_2 < \cdots$ such that (a) for any $m \ge m_0$, $x_m \notin B$, (b) $x_{m_k} \in A$ for every $k \ge 0$, and (c) for any $m$ such that $m_k < m < m_{k + 1}$, $x_m \notin A$.
	By assumption, we have $(x_{m_k}, x_{m_{k + 1}}) \in R$ for every $k \ge 0$, which contradicts the well-foundedness of $R$.
\end{proof}

Then, we can show that allowing history-dependent strategies does not increase the winning region for Verifier in $\WellFoundedGame{E}$.
\begin{lemma}\label{lem:history-dependent-well-founded-game-parity-game}
	If Verifier has a history-dependent winning strategy $\sigma$ in $\WellFoundedGame{E}$ from the position $(X_i, v; \bot, \dots, \bot)$, then $\sigma$ is also a winning strategy for Verifier in $\GameSemantics{E}$ from the position $(X_i, v)$.
\end{lemma}
\begin{proof}
	Let $R = \{ R_j \}_{j \in \{ j \mid \eta_j = \mu \}}$ be a family of well-founded relations that witness the winning condition in $\WellFoundedGame{E}$.
	Let $\rho$ be a maximal play in $\GameSemantics{E}$ consistent with $\sigma$.
	If $\rho$ is finite, then it ends in a Falsifier's position by the winning condition in $\WellFoundedGame{E}$.
	If $\rho$ is infinite, then it suffices to show that for each $j$ with $\eta_j = \mu$, $\rho$ satisfies the Streett condition $\{ (\mathcal{D}_{\le j}, \mathcal{D}_{< j}) \}$.
	Suppose that $\rho$ contains moves of the form~\eqref{eq:play-transitive-closure} for $j$ with $\eta_j = \mu$.
	By considering the corresponding play in $\WellFoundedGame{E}$, we have $(u_0, u_m) \in R_j$ since $\sigma$ and $R$ witness the winning condition in $\WellFoundedGame{E}$.
	By Lemma~\ref{lem:sequence-streett-condition-well-founded-relation}, this implies that $\rho$ satisfies the Streett condition $\{ (\mathcal{D}_{\le j}, \mathcal{D}_{< j}) \}$.
\end{proof}

\begin{theorem}\label{thm:well-founded-game-parity-game}
	Verifier wins from a position $(X_i, v)$ in $\GameSemantics{E}$ (Definition~\ref{def:game-semantics}) if and only if Verifier wins from a position $(X_i, v; \bot, \dots, \bot)$ in $\WellFoundedGame{E}$.
\end{theorem}
\begin{proof}
	The if part follows from Lemma~\ref{lem:history-dependent-well-founded-game-parity-game}.
	The only if part follows from Proposition~\ref{prop:strongly-memoryless-well-founded-game-parity-game}.
\end{proof}

\subsection{Translation of $\mu$CLP}

To solve a $\mu$CLP $E$, \citet{UnnoPOPL2023} proposed a $\mu$-to-$\nu$ transformation $\mathbf{elim}_{\mu}(E)$ using well-founded relations.
We show that finding a post-fixed point of $\mathbf{elim}_{\mu}(E)$ can be seen as finding a winning strategy for Verifier in the game $\WellFoundedGame{E}$ defined in Definition~\ref{def:well-founded-game}.

\begin{definition}[$\mathbf{elim}_{\mu}$ in \cite{UnnoPOPL2023}, slightly modified]
	Given a $\mu$CLP $E = \{ X_1(x_1) =_{\eta_1} \phi_1; \dots; X_n(x_n) =_{\eta_n} \phi_n \}$, a $\mu$CLP $\mathbf{elim}_{\mu}(E)$ is defined by replacing the $i$-th equation $X_i(x_i) =_{\eta_i} \phi_i$ in $E$ with the following equation for each $i = 1, \dots, n$.
	\[ \hat{X}_i(\hat{x}_1, \dots, \hat{x}_{i - 1}, x_i) \quad=_{\nu}\quad \phi_i[\lambda y. \psi_{i, 1}/X_1, \dots, \lambda y. \psi_{i, n}/X_n] \]
	where for each $j$, $\hat{x}_j$ is a fresh term variable whose sort $\sort(\hat{x}_j)$ extends that of $x_j$ with an additional element $\bot$, i.e., $\interpret{\sort(\hat{x}_j)} = \interpret{\sort(x_j)} \cup \{ \bot \}$; and $\psi_{i, j}$ is a formula defined as follows.
	\[ \psi_{i, j} \quad\coloneqq\quad \begin{cases}
		\hat{X}_{j}(\hat{x}_1, \dots, \hat{x}_{j - 1}, y) \land (\hat{x}_{j} \neq \bot \implies R_{j}(\hat{x}_{j}, y)) & j < i \\
		\hat{X}_{i}(\hat{x}_1, \dots, \hat{x}_{i - 1}, y) \land R_{i}(x_i, y) & j = i \\
		\hat{X}_{j}(\hat{x}_1, \dots, \hat{x}_{i - 1}, x_{i}, \bot, \dots, \bot, y) & j > i
	\end{cases} \]
	Here, $R_i \subseteq \mathcal{D}_i \times \mathcal{D}_i$ is a well-founded relation if $\eta_i = \mu$, and $R_i(x, y) = \True$ if $\eta_i = \nu$ (i.e.\ the conjunct involving $R_i$ can be omitted if $\eta_i = \nu$).
\end{definition}

\citet{UnnoPOPL2023} showed that if $(\hat{A}_1, \dots, \hat{A}_n)$ is a post-fixed point of $\mathbf{elim}_{\mu}(E)$, then $x_i \in X_i^{*}$ holds for each $i$ and $x_i$ such that $(\bot, \dots, \bot, x_i) \in \hat{A}_i$.
We interpret this result in terms of the winning strategy for Verifier in $\WellFoundedGame{E}$.

\begin{remark}
	For presentation purposes, we slightly modified the original transformation in~\cite{UnnoPOPL2023}, but the essence is the same.
	Specifically, we do not actually need to introduce additional arguments $\hat{x}_j$ and relations $W_j$ for $\nu$-equations.
	The original transformation in \cite{UnnoPOPL2023} is defined so that it does not introduce such redundant arguments and relations.
	Also, the sort of $\hat{x}_j$ is defined as the product $\mathtt{sort}(x_j) \times \mathtt{bool}$ of the sort of $x_j$ and the sort of boolean values in \cite{UnnoPOPL2023} so that a value $v$ of sort $\mathtt{sort}(x_j)$ can be encoded as $(v, \True)$ and $\bot$ can be encoded as $({*}, \False)$ where ${*}$ is an arbitrary value.
\end{remark}

\begin{theorem}\label{thm:well-founded-strategy}
	Let $(\hat{\mathcal{A}}_1, \dots, \hat{\mathcal{A}}_n)$ be a post-fixed point of $\mathbf{elim}_{\mu}(E)$.
	Then, there exists a memoryless (but not strongly memoryless) winning strategy for Verifier in $\WellFoundedGame{E}$ from any position $(X_i, v; h_1, \dots, h_{i - 1})$ such that $(h_1, \dots, h_{i - 1}, v) \in \hat{\mathcal{A}}_i$.
	In particular, if $(\bot, \dots, \bot, v) \in \hat{\mathcal{A}}_i$, then the solution of $E$ satisfies $v \in X_i^{*}$.
\end{theorem}
\begin{proof}
	The latter statement follows from Theorem~\ref{thm:well-founded-game-parity-game} and~\ref{thm:game-semantics}, so we focus on the former statement.
	We define a strategy for Verifier in $\WellFoundedGame{E}$ as follows.
	If the current position is $(X_i, v; h_1, \dots, h_{i - 1})$ such that $(h_1, \dots, h_{i - 1}, v) \in \hat{\mathcal{A}}_i$
	then since $(\hat{\mathcal{A}}_1, \dots, \hat{\mathcal{A}}_n)$ is a post-fixed point of $\mathbf{elim}_{\mu}(E)$,
	\[ \interpret{\phi_i[\lambda y. \psi_{i, 1}/X_1, \dots, \lambda y. \psi_{i, n}/X_n]}(\hat{\mathcal{A}}_1, \dots, \hat{\mathcal{A}}_n)(h_1, \dots, h_{i - 1}, v) = \interpret{\phi_i}(\hat{\mathcal{A}'}_1, \dots, \hat{\mathcal{A}'}_n)(v) = \True \]
	where $\hat{\mathcal{A}'}_1, \dots, \hat{\mathcal{A}'}_n$ are defined as follows.
	\[ \hat{\mathcal{A}'}_j \quad\coloneqq\quad \begin{cases}
		\{ u \mid (h_1, \dots, h_{j-1}, u) \in \hat{\mathcal{A}}_j \land (h_j \neq \bot \implies (h_j, u) \in R_j) \} & j < i \\
		\{ u \mid (h_1, \dots, h_{i-1}, u) \in \hat{\mathcal{A}}_i \land (v, u) \in R_i \} & j = i \\
		\{ u \mid (h_1, \dots, h_{i-1}, v, \bot, \dots, \bot, u) \in \hat{\mathcal{A}}_j \} & j > i
	\end{cases} \]
	Therefore, we define the next move of Verifier to be $(\hat{\mathcal{A}'}_1, \dots, \hat{\mathcal{A}'}_n; h_1, \dots, h_{i - 1}, v, \bot, \dots, \bot)$.
	\iflong

	It remains to show that this strategy is winning.
	For each Falsifier's response to the above strategy of Verifier:
	\[ (X_i, v; h_1, \dots, h_{i - 1}) \to (\hat{\mathcal{A}'}_1, \dots, \hat{\mathcal{A}'}_n; h_1, \dots, h_{i - 1}, v, \bot, \dots, \bot) \to (X_j, u; h'_1, \dots, h'_{j-1}) \]
	if $h_j \neq \bot$, then $j$ must be less than or equal to $i$.
	If $j < i$ and $h'_j = h_j \neq \bot$, then $(h_j, u) \in R_j$ by the definition of $\hat{\mathcal{A}'}_j$.
	If $j = i$, then $(v, u) \in R_i$.
	In either case, we have $(h'_1, \dots, h'_{j-1}, u) \in \hat{\mathcal{A}}_j$ again.
	Therefore, any maximal play $\rho$ consistent with the strategy satisfies the winning condition for Verifier in $\WellFoundedGame{E}$.
	\else
	It is straightforward to show that this strategy is winning.
	\fi
\end{proof}

Now, we show the relationship between $E[R]$ in Section~\ref{sec:validity-checking} and $\mathbf{elim}_{\mu}(E)$.
Recall that a post-fixed point of $E[R]$ induces a memoryless winning strategy for Verifier in $\GameSemantics{E}$ (Theorem~\ref{thm:soundness}) and that a memoryless winning strategy for Verifier in $\GameSemantics{E}$ corresponds to a strongly memoryless winning strategy for Verifier in $\WellFoundedGame{E}$ (Proposition~\ref{prop:strongly-memoryless-well-founded-game-parity-game}).
Thus, it is natural to ask whether a post-fixed point of $E[R]$ can be used to construct a post-fixed point of $\mathbf{elim}_{\mu}(E)$.
The following theorem answers this question in the affirmative.
Since $E[R]$ is complete (Theorem~\ref{thm:completeness}), this theorem also gives another proof of completeness of $\mathbf{elim}_{\mu}(E)$ \cite{UnnoPOPL2023}.

\begin{theorem}\label{thm:parity-relation-well-founded-relation}
	Let $R = \{ R_{i, j} \}_{i, j = 1, \dots, n}$ be a parity relation on $\mathcal{D}$ and $(\mathcal{A}_1, \dots, \mathcal{A}_n)$ be a post-fixed point of $E[R]$ (Definition~\ref{def:parity-relation-transformation}).
	Then, there exists a post-fixed point $(\hat{\mathcal{A}}_1, \dots, \hat{\mathcal{A}}_n)$ of $\mathbf{elim}_{\mu}(E)$ such that for any $x \in \mathcal{A}_i$, we have $(\bot, \dots, \bot, x) \in \hat{\mathcal{A}}_i$.
	Here, we use $\{ \hat{R}_{\mathcal{D}_{\le j}, \mathcal{D}_{< j}} \mid \eta_j = \mu \}$ as well-founded relations in $\mathbf{elim}_{\mu}(E)$ (cf.\ Lemma~\ref{lem:streett-relation-well-founded-relation}).
\end{theorem}
\begin{proof}
	See
	\iflong
	Appendix~\ref{sec:proofs}
	\else
	Appendix~B
	\fi
	for the proof.
\end{proof}

%% file: sas19.tex
\section{Game Interpretation of the Transformation in \cite{KobayashiSAS2019}}
\label{sec:comparison-sas19}

We show how the $\mu$-to-$\nu$ transformation proposed by \citet{KobayashiSAS2019} can be interpreted in terms of game semantics.

Let $E = \{ X_1(x_1) =_{\eta_1} \phi_1; \dots; X_n(x_n) =_{\eta_n} \phi_n \}$ a $\mu$CLP.
In this section, we assume that $n = 2 n'$ is even and $\eta_i = \mu$ if and only if $i$ is odd without loss of generality.
\citet{KobayashiSAS2019} proposed the following $\mu$-to-$\nu$ transformation.

\begin{definition}
	We define a $\mu$CLP $E^{\mathsf{cnt}}$ by applying the following transformation to each equation in $E$.
	\begin{align*}
		X_{2 k}(x_{2 k}) =_{\eta_{2 k}} \phi_{2 k} \quad&\mapsto\quad \dot{X}_{2 k}(c_1, c_3, \dots, c_{2 n' - 1}, x_{2 k}) =_{\nu} \phi_{2 k}[\Theta_{2 k}] \\
		X_{2 k - 1}(x_{2 k - 1}) =_{\eta_{2 k - 1}} \phi_{2 k - 1} \quad&\mapsto\quad \dot{X}_{2 k - 1}(c_1, c_3, \dots, c_{2 n' - 1}, x_{2 k - 1}) =_{\nu} c_{2 k - 1} > 0 \land \phi_{2 k - 1}[\Theta_{2 k - 1}]
	\end{align*}
	where $\Theta_i$ is a substitution for predicate variables defined as follows.
	\[ \Theta_{2 k}(X_j) \coloneqq \dot{X}_j(c_1, \dots, c_{2 n' - 1}, {-}) \qquad \Theta_{2 k - 1}(X_j) \coloneqq \dot{X}_j(c_1, \dots, c_{2 k - 1} - 1, \dots, c_{2 n' - 1}, {-}) \]
	That is, when we unfold the fixed-point equation for a $\mu$-variable $\dot{X}_{2 k - 1}$, the corresponding counter $c_{2 k - 1}$ is decremented by one on the right-hand side of the equation, while the other counters remain unchanged.
\end{definition}

This transformation satisfies the property that $(c_1, \dots, c_{2 n' - 1}, x_i) \in \dot{X}^{*}_i$ implies $x_i \in X^{*}_i$ for each $i$ and $(c_1, \dots, c_{2 n' - 1}, x_i)$.
We interpret this result in terms of game semantics.

\begin{definition}\label{def:counter-game}
	Given a $\mu$CLP $E$, we define a game $\CounterGame{E}$ played by Verifier (Player 1) and Falsifier (Player 2) as follows.
	\begin{itemize}
		\item Verifier's positions: $(X_i, v; c_1, \dots, c_{2 n' - 1})$ where $v \in \mathcal{D}_i$ and $c_j \in \mathbb{N}$ for each $j$.
		\item Falsifier's positions: a tuple $(\mathcal{A}_1, \dots, \mathcal{A}_n; c_1, \dots, c_{2 n' - 1})$ where $\mathcal{A}_j \subseteq \mathcal{D}_j$ for each $j$.
		\item Verifier's moves: Let $(\mathcal{A}_1, \dots, \mathcal{A}_n)$ be such that $\interpret{\phi_i}(\mathcal{A}_1, \dots, \mathcal{A}_n)(v)$ is true.
		\begin{align*}
			(X_{2 k - 1}, v; c_1, \dots, c_{2 n' - 1}) &\to (\mathcal{A}_1, \dots, \mathcal{A}_n; c_1, \dots, c_{2 k - 1} - 1, \dots, c_{2 n' - 1}) \qquad \text{if $c_{2 k - 1} > 0$} \\
			(X_{2 k}, v; c_1, \dots, c_{2 n' - 1}) &\to (\mathcal{A}_1, \dots, \mathcal{A}_n; c_1, \dots, c_{2 n' - 1})
		\end{align*}
		\item Falsifier's moves: $(\mathcal{A}_1, \dots, \mathcal{A}_n; c_1, \dots, c_{2 n' - 1}) \to (X_j, v; c_1, \dots, c_{2 n' - 1})$ such that $v \in \mathcal{A}_j$.
	\end{itemize}
	A maximal play is winning for Verifier if (a) it is finite and ends in a Falsifier's position or (b) it is infinite.
\end{definition}

It is straightforward to see that if Verifier is winning from $(X_i, v; c_1, \dots, c_{2 n' - 1})$ in $\CounterGame{E}$, then $(X_i, v)$ is also winning in the game semantics $\GameSemantics{E}$ defined in Definition~\ref{def:game-semantics}.

\begin{proposition}
	If Verifier has a winning strategy $\sigma$ from $(X_i, v; c_1, \dots, c_{2 n' - 1})$ in $\CounterGame{E}$, then Verifier also has a winning strategy $\sigma'$ from $(X_i, v)$ in $\GameSemantics{E}$.
\end{proposition}
\begin{proof}
	We have an injective but not surjective mapping from $\Play[(X_i, v; c_1, \dots, c_{2 n' - 1})]{\CounterGame{E}}$ to $\Play[(X_i, v)]{\GameSemantics{E}}$ by forgetting the counters.
	Let $P \subseteq \Play[(X_i, v)]{\GameSemantics{E}}$ be the image of this mapping.
	Then, $P$ is closed under appending any move by Falsifier.

	We define a (history-dependent) strategy $\sigma'$ as follows.
	If a play $\rho \in \Play[(X_i, v)]{\GameSemantics{E}}$ ending in a Verifier's position is in $P$, then we define $\sigma'(\rho)$ by forgetting the counters in $\sigma(\rho')$.
	Otherwise, we define $\sigma'(\rho)$ arbitrarily.
	Then, it is straightforward to see that $\StrategyPlay[(X_i, v)]{\sigma'}{\GameSemantics{E}} \subseteq P$.

	For any finite maximal play $\rho \in \StrategyPlay[(X_i, v)]{\sigma'}{\GameSemantics{E}}$, since the corresponding play in $\CounterGame{E}$ is winning for Verifier, $\rho$ must end in a Falsifier's position.
	For any infinite play $\rho \in \StrategyPlay[(X_i, v)]{\sigma'}{\GameSemantics{E}}$, by definition of $\CounterGame{E}$, $X_{2 k - 1}$ can occur at most $c_{2 k - 1} + 1$ times in $\rho$ for each $k = 1, \dots, n'$.
	Therefore, $\rho$ satisfies the parity condition, and hence it is winning for Verifier.
\end{proof}

\begin{theorem}
	Let $(\dot{\mathcal{A}}_1, \dots, \dot{\mathcal{A}}_n)$ is a post-fixed point of $E^{\mathsf{cnt}}$.
	Then, there exists a (memoryless) strategy for Verifier in $\CounterGame{E}$ that witnesses that $(X_i, v; c_1, \dots, c_{2 n' - 1})$ is winning for each $i$ and $(c_1, \dots, c_{2 n' - 1}, v) \in \dot{\mathcal{A}}_i$.
	In particular, if $(c_1, \dots, c_{2 n' - 1}, v) \in \dot{\mathcal{A}}_i$, then the solution of $E$ satisfies $v \in X_i^{*}$.
\end{theorem}
\begin{proof}
	We define a memoryless strategy $\sigma$ as follows.
	Let $(X_i, v; c_1, \dots, c_{2 n' - 1})$ be a Verifier's position in $\CounterGame{E}$ such that $(c_1, \dots, c_{2 n' - 1}, v) \in \dot{\mathcal{A}}_i$.
	\begin{itemize}
		\item If $i = 2 k - 1$ is odd, then by definition of $E^{\mathsf{cnt}}$, $c_{2 k - 1} > 0$ and $\interpret{\phi_i}(\mathcal{A}_1, \dots, \mathcal{A}_n)(v)$ is true where $\mathcal{A}_j \coloneqq \{ x_i \mid (c_1, \dots, c_{2 k - 1} - 1, \dots, c_{2 n' - 1}, v) \in \dot{\mathcal{A}}_j \}$.
		Hence, we define the next move of $\sigma$ to be $(\mathcal{A}_1, \dots, \mathcal{A}_n; c_1, \dots, c_{2 k - 1} - 1, \dots, c_{2 n' - 1})$.
		\item If $i = 2 k$ is even, then $\interpret{\phi_i}(\mathcal{A}_1, \dots, \mathcal{A}_n)(v)$ is true where $\mathcal{A}_j \coloneqq \{ x_i \mid (c_1, \dots, c_{2 n' - 1}, v) \in \dot{\mathcal{A}}_j \}$.
		Hence, we define the next move of $\sigma$ to be $(\mathcal{A}_1, \dots, \mathcal{A}_n; c_1, \dots, c_{2 n' - 1})$.
	\end{itemize}
	In either case, for any Falsifier's response $(X_j, u; c'_1, \dots, c'_{2 n' - 1})$, we have $(c'_1, \dots, c'_{2 n' - 1}, u) \in \dot{\mathcal{A}}_j$.
	Therefore, Verifier never gets stuck, and hence $\sigma$ is a winning strategy.
\end{proof}

%% file: popl25.tex
\section{Comparison of \cite{TsukadaPOPL2025} with the Parity Game Semantics of $\mu$CLPs}
\label{sec:comparison-popl25}

Recall that the transformation in \cite{UnnoPOPL2023} uses the relationship between Streett relations and well-founded relations (Lemma~\ref{lem:streett-relation-well-founded-relation}).
In this section, we consider a similar relationship between disjunctively well-founded relations \cite{PodelskiLICS2004} and Streett relations.
Then, this explains the relationship between the game semantics of $\mu$CLPs and the game introduced in \cite{TsukadaPOPL2025}.

\begin{definition}
	A \emph{disjunctively well-founded relation} $R$ is a finite union $R = R_1 \cup \dots \cup R_k$ of well-founded relations $R_1, \dots, R_k$.
\end{definition}

\begin{lemma}\label{lem:streett-relation-disjunctively-well-founded-relation}
	A relation $R$ is $(A, B)$-Streett if and only if the relation $\check{R}_{A, B}$ defined below is disjunctively well-founded:
	\[ \check{R}_{A, B} \quad\coloneqq\quad \{ (x_0, x_m) \mid m > 0 \land x_0 \mathrel{R} x_1 \mathrel{R} \cdots \mathrel{R} x_{m-1} \mathrel{R} x_m \land x_0, x_m \in A \land x_0, x_1 \dots, x_{m} \notin B \} \]
\end{lemma}
\begin{proof}
	By Lemma~\ref{lem:streett-relation-well-founded-relation}, $R$ is $(A, B)$-Streett if and only if $\hat{R}_{A, B}$ is well-founded.
	By \cite[Theorem~1]{PodelskiLICS2004}, a relation $Q$ is well-founded if and only if its transitive closure $Q^{+}$ is disjunctively well-founded.
	Since $(\hat{R}_{A, B})^{+} = \check{R}_{A, B}$ holds, $\hat{R}_{A, B}$ is well-founded if and only if $\check{R}_{A, B}$ is disjunctively well-founded.
\end{proof}

Similarly to Corollary~\ref{cor:winning-condition-well-founded-relation}, we can rephrase the winning condition for Verifier in the game semantics of $\mu$CLP in terms of disjunctively well-founded relations.

\begin{proposition}\label{prop:winning-condition-disjunctively-well-founded-relation}
	Verifier is winning from $(X_i, v)$ in the game semantics of $\mu$CLP $E$ if and only if there exist a memoryless strategy $\sigma$ and a family of disjunctively well-founded relations $R = \{ R_j \subseteq \mathcal{D}_j \times \mathcal{D}_j \}_{j \in \{ j \mid \eta_j = \mu \}} $ such that (a) for any play of the form~\eqref{eq:play-transitive-closure-non-strict}, if it is consistent with $\sigma$, then $(u_0, u_m) \in R_j$ and (b) any finite maximal play consistent with $\sigma$ ends at a Falsifier's position.
	\begin{align}
		&(X_i, v) \to && \cdots &&\to (X_{j_0}, u_0) && \qquad j_0 = j \\
		&&&\to (\mathcal{A}_{1, 1}, \dots, \mathcal{A}_{1, n}) &&\to (X_{j_1}, u_1) && \qquad j_1 \ge j \\
		&&&\qquad\vdots \\
		&&&\to (\mathcal{A}_{m-1, 1}, \dots, \mathcal{A}_{m-1, n}) &&\to (X_{j_{m-1}}, u_{m-1}) && \qquad j_{m-1} \ge j \\
		&&&\to (\mathcal{A}_{m, 1}, \dots, \mathcal{A}_{m, n}) &&\to (X_{j_m}, u_m) && \qquad j_m = j \label{eq:play-transitive-closure-non-strict}
	\end{align}
	Note that the play~\eqref{eq:play-transitive-closure-non-strict} requires non-strict inequalities on $j_1, \dots, j_{m-1}$ whereas the play~\eqref{eq:play-transitive-closure} in Corollary~\ref{cor:winning-condition-well-founded-relation} requires strict inequalities.
	\qed
\end{proposition}

Now, we define a game based on disjunctively well-founded relations.
The definition is similar to the game in Definition~\ref{def:well-founded-game} based on well-founded relations, but remembering only the last visit to $(X_i, v)$ is not enough in this case.
Instead, we need to remember all the previous visits to $(X_i, v)$.

\begin{definition}\label{def:disjunctively-well-founded-game}
	Given a $\mu$CLP $E = \{ X_1 =_{\eta_1} \phi_1; \dots; X_n =_{\eta_n} \phi_n \}$, we define a game $\DisjunctivelyWellFoundedGame{E}$ played by Verifier (Player 1) and Falsifier (Player 2) as follows.
	\begin{itemize}
		\item Verifier's positions: $(X_i, v; h_1, \dots, h_i)$ where $v \in \mathcal{D}_i$ and $h_j \subseteq_{\mathrm{fin}} \mathcal{D}_j$ for each $j$.
		Here, $\subseteq_{\mathrm{fin}}$ denotes the finite subset relation.
		\item Falsifier's positions: $(\mathcal{A}_1, \dots, \mathcal{A}_n; h_1, \dots, h_n)$ where $\mathcal{A}_j \subseteq \mathcal{D}_j$ and $h_j \subseteq_{\mathrm{fin}} \mathcal{D}_j$ for each $j$.
		\item Verifier's moves: $(X_i, v; h_1, \dots, h_i) \to (\mathcal{A}_1, \dots, \mathcal{A}_n; h_1, \dots, h_i \cup \{ v \}, \emptyset, \dots, \emptyset)$ such that $\interpret{\phi_i}(\mathcal{A}_1, \dots, \mathcal{A}_n)(v)$ is true.
		\item Falsifier's moves: $(\mathcal{A}_1, \dots, \mathcal{A}_n; h_1, \dots, h_n) \to (X_j, v; h_1, \dots, h_j)$ such that $v \in \mathcal{A}_j$.
	\end{itemize}
\end{definition}
\begin{definition}[winning condition]
	Verifier wins from a position $(X_i, v; h_1, \dots, h_i)$ if there exist a family of disjunctively well-founded relations $\{ R_j \}_{j \in \{ j \mid \eta_j = \mu \}}$ and a strategy $\sigma$ for Verifier such that for each maximal play $\rho$ consistent with the strategy $\sigma$,
	\begin{itemize}
		\item if $\rho$ is finite, then it ends in a Falsifier's position (i.e., Verifier never gets stuck), and
		\item if $\rho$ is infinite and $\eta_j = \mu$, then for each Verifier's move
		\[ (X_j, v; h_1, \dots, h_j) \quad\to\quad (A_1, \dots, A_n; h_1, \dots, h_j \cup \{ v \}, \emptyset, \dots, \emptyset) \qquad \text{in $\rho$} \]
		and for any $u \in h_j$, we have $(u, v) \in R_j$.
	\end{itemize}
\end{definition}

The game defined above is almost the same as the game considered in \cite[Section~7.2]{TsukadaPOPL2025}.
The main difference is that Verifier's moves in \cite{TsukadaPOPL2025} are defined as $(X_i, v; h_1, \dots, h_i) \to (\mathcal{A}_1, \dots, \mathcal{A}_n; h_1, \dots, h_i \cup \{ v \}, h_{i+1}, \dots, h_n)$ without resetting $h_{i+1}, \dots, h_n$ to $\emptyset$, which makes the winning condition for Verifier unnecessarily stronger (i.e.\ it is harder for Verifier to win).

In what follows, we fix an initial position in $\DisjunctivelyWellFoundedGame{E}$ as $(X_i, v; \emptyset, \dots, \emptyset)$.
Similarly to the case in Section~\ref{sec:comparison-popl23}, it is straightforward to see that the mapping $\Play[(X_i, v; \emptyset, \dots, \emptyset)]{\DisjunctivelyWellFoundedGame{E}} \to \Play[(X_i, v)]{\GameSemantics{E}}$ that removes $h_1, \dots, h_n$ from each position in a play is a bijection, which further induces a bijection between strategies for $\DisjunctivelyWellFoundedGame{E}$ and those for $\GameSemantics{E}$.

For any play $\rho \in \Play[(X_i, v)]{\GameSemantics{E}}$ of the form~\eqref{eq:play-transitive-closure-non-strict}, the corresponding play $\rho' \in \Play[(X_i, v; \emptyset, \dots, \emptyset)]{\DisjunctivelyWellFoundedGame{E}}$ contains the position $(X_j, u_m; h_1, \dots, h_j)$ at the end of the play such that $u_0 \in h_j$.
Conversely, for any play $\rho' \in \Play[(X_i, v; \emptyset, \dots, \emptyset)]{\DisjunctivelyWellFoundedGame{E}}$ ending with a position $(X_j, u; h_1, \dots, h_j)$ such that $w \in h_j$, the corresponding play $\rho \in \Play[(X_i, v)]{\GameSemantics{E}}$ has the form~\eqref{eq:play-transitive-closure-non-strict} such that $w = u_0$ and $u = u_m$.
Therefore, Proposition~\ref{prop:winning-condition-disjunctively-well-founded-relation} can be rephrased as the following correspondence between (memoryless) winning strategies for $\GameSemantics{E}$ and strongly memoryless winning strategies for $\DisjunctivelyWellFoundedGame{E}$.
Here, we say a strategy for $\DisjunctivelyWellFoundedGame{E}$ is \emph{strongly memoryless} if it is memoryless and does not depend on $h_1, \dots, h_n$.

\begin{proposition}\label{prop:strongly-memoryless-disjunctively-well-founded-game-parity-game}
	Verifier wins from a position $(X_i, v)$ in $\GameSemantics{E}$ (Definition~\ref{def:game-semantics}) if and only if Verifier wins from a position $(X_i, v; \emptyset, \dots, \emptyset)$ in $\DisjunctivelyWellFoundedGame{E}$ with a strongly memoryless strategy.
	\qed
\end{proposition}

Similarly to the case in Section~\ref{sec:comparison-popl23} (Theorem~\ref{thm:well-founded-game-parity-game}), we can also show that allowing history-dependent strategies does not increase the winning region for Verifier in $\DisjunctivelyWellFoundedGame{E}$.

\begin{theorem}\label{thm:history-dependent-disjunctively-well-founded-game-parity-game}
	Verifier wins from a position $(X_i, v)$ in $\GameSemantics{E}$ (Definition~\ref{def:game-semantics}) if and only if Verifier wins from a position $(X_i, v; \emptyset, \dots, \emptyset)$ in $\DisjunctivelyWellFoundedGame{E}$.
	\iflong\else
	\qed
	\fi
\end{theorem}
\iflong
\begin{proof}
	It suffices to show that if Verifier in $\DisjunctivelyWellFoundedGame{E}$ has a (history-dependent) winning strategy $\sigma$ from $(X_i, v; \emptyset, \dots, \emptyset)$, then $\sigma$ is also a winning strategy for Verifier in $\GameSemantics{E}$ from $(X_i, v)$.
	Let $\{ R_j \}_{j \in \{ j \mid \eta_j = \mu \}}$ be a family of disjunctively well-founded relations witnessing that $\sigma$ is a winning strategy.
	For any maximal play $\rho$ in $\GameSemantics{E}$ consistent with $\sigma$ and starting from $(X_i, v)$, since the corresponding play in $\DisjunctivelyWellFoundedGame{E}$ satisfies the winning condition, if $\rho$ is finite, then it ends in a Falsifier's position.
	Otherwise, if $\rho$ is infinite, then for any $j$ such that $\eta_j = \mu$, and for any prefix of $\rho$ of the form~\eqref{eq:play-transitive-closure-non-strict}, we have $(u_0, u_m) \in R_j$.
	By Lemma~\ref{lem:sequence-streett-condition-disjunctively-well-founded-relation}, this implies that $\rho$ satisfies the Streett condition $\{ (\mathcal{D}_{\le j}, \mathcal{D}_{< j}) \}$ for each $j$ such that $\eta_j = \mu$, which is equivalent to the parity condition for $\rho$.
	Therefore, $\sigma$ is a winning strategy for Verifier in $\GameSemantics{E}$ from $(X_i, v)$.
\end{proof}
\fi

Theorem~\ref{thm:history-dependent-disjunctively-well-founded-game-parity-game} is proved using the following lemma (cf.\ Lemma~\ref{lem:sequence-streett-condition-well-founded-relation}).
\begin{lemma}\label{lem:sequence-streett-condition-disjunctively-well-founded-relation}
	Let $R$ be a disjunctively well-founded relation on $\mathcal{D}$ and $x_0 x_1 \cdots \in \mathcal{D}^{\omega}$ be an infinite sequence.
	Assume that for any $l$ and $m$ with $l < m$, if $x_l, x_m \in A$ and $x_l, x_{l + 1}, \dots, x_m \notin B$, then $(x_l, x_m) \in R$.
	Then, the sequence satisfies the Streett condition $\{ (A, B) \}$.
\end{lemma}
\iflong
\begin{proof}
	We use Ramsey's theorem, similarly to the proof of \cite[Theorem~1]{PodelskiLICS2004}.
	Let $R_1, \dots, R_k$ be well-founded relations such that $R = R_1 \cup \cdots \cup R_k$.
	Suppose that $x_0 x_1 \dots \in T$ violates the Streett condition.
	Then, there exists an infinite sequence $m_0 < m_1 < m_2 < \cdots$ such that (a) for any $m \ge m_0$, $x_m \notin B$ and (b) $x_{m_k} \in A$ for every $k \ge 0$.
	By assumption, for any pair $(k, k')$ with $k < k'$, there exists $i \in \{ 1, \dots, k \}$ such that $(x_{m_k}, x_{m_{k'}}) \in R_i$.
	By Ramsey's theorem, there exists an infinite subsequence $m_{k_0} < m_{k_1} < m_{k_2} < \cdots$ and $i \in \{ 1, \dots, k \}$ such that $(x_{m_{k_l}}, x_{m_{k_{l'}}}) \in R_i$ for any $l < l'$.
	This contradicts the well-foundedness of $R_i$.
\end{proof}
\else
\begin{proof}[Proof Sketch]
	Similarly to the proof of \cite[Theorem~1]{PodelskiLICS2004}, this is proved by applying Ramsey's theorem.
\end{proof}
\fi

%% file: implementation.tex
\section{Efficient Reduction to Constraint Solving}\label{sec:implementation}

In this section, we explain how to reduce the validity checking problem of a $\mu$CLP with a closed query to a constraint solving problem, using the transformation $E[R]$ in Section~\ref{sec:validity-checking}.
Then, we briefly describe techniques to decrease the size of the generated constraints.

\subsection{Reducing $\mu$CLPs to Constraint Solving Problems}

Suppose that a $\mu$CLP $E$ with a closed query $\psi$ is given.
To show the validity $\models (E, \psi)$ (Definition~\ref{def:muCLP-with-query}), it suffices to find an assignment $X_i(\tilde{x}_i) \mapsto \mathcal{A}_i \subseteq \mathcal{D}_i$ to each predicate variable $X_i$ such that the assignment satisfies the query $\psi$ and is a post-fixed point of $E[R]$ for some parity relation $R$ by Lemma~\ref{lem:lower-bound-query} and Theorem~\ref{thm:soundness} and~\ref{thm:completeness}:
\begin{equation}
	{\mathcal{A}}_1, \dots, \mathcal{A}_n \quad\models\quad \psi \land \bigwedge_{i = 1}^n (X_i(\tilde{x}_i) \implies \phi'_i)
	\label{eq:constraint-solving-muCLP}
\end{equation}
where for each $i$, $\phi'_i$ is the right-hand side of the $i$-th equation in $E[R]$:
\[ \phi'_i \quad\coloneqq\quad \phi_i[\lambda \tilde{y}_1. X_1(\tilde{y}_1) \land R_{i, 1}(\tilde{x}_i, \tilde{y}_1)/X_1, \dots, \lambda \tilde{y}_n. X_n(\tilde{y}_n) \land R_{i, n}(\tilde{x}_i, \tilde{y}_n)/X_n]. \]
This problem can be naturally captured as a constraint solving problem for a slight extension of pfwCSP~\cite{UnnoCAV2021}, which we explain below.

A \emph{pCSP}~\cite{SatakeAAAI2020} is a finite set of constraints over predicate variables of the form $\phi \lor \bigvee_{j = 0}^l X_{i_j}(\tilde{t}_j) \lor \bigvee_{k = l + 1}^m \lnot X_{i_k}(\tilde{t}_k)$ where $0 \le l \le m$.
A pCSP is \emph{valid} if there exists an assignment $X_i(\tilde{x}_i) \mapsto \mathcal{A}_i \subseteq \mathcal{D}_i$ to each predicate variable $X_i$ such that the assignment satisfies all the constraints.
Note that pCSPs are a generalization of constrained Horn clauses (CHCs): a pCSP is CHCs (co-CHCs) if each constraint contains at most one positive (negative) occurrence of predicate variables.
A \emph{pfwCSP}~\cite{UnnoCAV2021} is an extension of pCSPs where predicate variables are partitioned into three disjoint sets: \emph{ordinary predicate variables}, \emph{functional predicate variables}, and \emph{well-founded predicate variables}.
Ordinary predicate variables are those in pCSPs.
A functional predicate variable $F(\tilde{x}, y)$ is a predicate variable that must be instantiated with a functional relation, i.e., for any $\tilde{x}$, there must be exactly one $y$ such that $F(\tilde{x}, y)$ holds.
A well-founded predicate variable $W(\tilde{x}, \tilde{y})$ is a predicate variable that must be instantiated with a well-founded relation, i.e., there is no infinite sequence $\tilde{x}_0, \tilde{x}_1, \dots$ such that $W(\tilde{x}_k, \tilde{x}_{k + 1})$ holds for any $k$.

For our purpose, we slightly extend pfwCSPs by adding a new kind of predicate variables called \emph{parity predicate variables}, which is a family of predicate variables $\{ P_i(\tilde{x}, \tilde{y}) \}_{i = 1, \dots, n}$ that must be instantiated with a parity relation (Definition~\ref{def:parity-relation}).
Then, it is straightforward to encode the problem~\eqref{eq:constraint-solving-muCLP} as a pfwCSP with parity predicate variables.
Since we assume positive occurrences of predicate variables in $\mu$CLPs, the resulting pfwCSP constraints are of co-CHC form.
Here, we comment on how to encode quantifiers.
Handling universal quantifiers does not cause any problem by definition of the validity of p(fw)CSPs.
Existential quantifiers can be eliminated by Skolemization using functional predicate variables, i.e., replacing $\exists y. \phi$ with $\forall y. F(\tilde{x}, y) \implies \phi$ where $F$ is a fresh functional predicate variable and $\tilde{x}$ are the term variables visible at $\exists y. \phi$.

\subsection{Dependency Analysis and Priority Reassignment}\label{sec:dependency-analysis}

After encoding the problem as a pfwCSP with parity predicate variables, we solve it using the template-based pfwCSP solver \PCSat~\cite{UnnoCAV2021}.
A straightforward way to support parity predicate variables in \PCSat\ is to use the parity relation~\eqref{eq:on-exit-progress-measure-relation-disjoint-union} as a template.
In this approach, the solver needs to synthesize functions $\{r_i : \mathcal{D}_i \rightharpoonup \mathbf{Ord}^n\}_{i=1,\dots,n}$ such that the induced parity relation $R$ satisfies the given constraints.
However, this approach may impose a significant burden on the solver.
It requires synthesizing $n$ functions $r_i : \mathcal{D}_i \rightharpoonup \mathbf{Ord}^n$ with $n$ components, where $n$ is the number of equations in the input $\mu$CLP.
Moreover, each component $r_{i,j} : \mathcal{D}_i \rightharpoonup \mathbf{Ord}$ is typically represented as a function $r_{i,j} : \mathcal{D}_i \to \mathbb{N}^m$ for some $m$, where $\mathbb{N}^m$ is equipped with the lexicographic order.
To mitigate this issue, we present a more efficient template obtained by analysing variable dependencies in the input $\mu$CLP.

One simple way to reduce the number of components of the functions $r_i$ is to assign the same priority to consecutive $\mu$-equations or consecutive $\nu$-equations.
For example, consider a $\mu$CLP $E = \{ X_1 =_{\mu} \phi_1; X_2 =_{\mu} \phi_2; X_3 =_{\nu} \phi_3; X_4 =_{\nu} \phi_4 \}$.
We can assign the same priority to $X_1$ and $X_2$, and likewise to $X_3$ and $X_4$.
As a result, the functions to be synthesized become $\{r_i : \mathcal{D}_i \rightharpoonup \mathbf{Ord}^2\}_{i=1, 2, 3, 4}$, instead of $\{r_i : \mathcal{D}_i \rightharpoonup \mathbf{Ord}^4\}_{i=1, 2, 3, 4}$.

There is a more sophisticated way.
Let $E = \{ X_1 =_{\eta_1} \phi_1; \dots; X_n =_{\eta_n} \phi_n \}$ be a $\mu$CLP.
It is often the case that the set of free predicate variables $\FreePredVars(\phi_i)$ in the right-hand side of the $i$-th equation is a proper subset of $\{ X_1, \dots, X_n \}$.
By exploiting this fact, we can reassign priorities to the equations to reduce the number of components of the functions $r_i : \mathcal{D}_i \rightharpoonup \mathbf{Ord}^n$.
To this end, we adapt the notion of variable dependency graphs (e.g.~\cite[Section~3.3]{NeeleLMCS2024}) and alternation depth (e.g.~\cite{Gradel2002}) to our setting as follows.

\begin{definition}[variable dependency graph]
	Let $E = \{ X_1 =_{\eta_1} \phi_1; \dots; X_n =_{\eta_n} \phi_n \}$ be a $\mu$CLP.
	The \emph{variable dependency graph} of $E$ is a vertex-labelled graph $(V, F, \eta)$ where $V = \{ 1, \dots, n \}$ is the set of indices, and $F = \{ (i, j) \mid i = 1, \dots, n; X_j \in \FreePredVars(\phi_i) \} \subseteq V \times V$.
	Each vertex $i \in V$ is labelled with $\eta_i \in \{ \mu, \nu \}$.
\end{definition}

\begin{definition}[alternation depth]\label{def:alternation-depth}
	We say that a vertex $j$ is \emph{$k$-reachable} from $i$ (denoted $i \xrightarrow[> k]{}^{+} j$) if there exists a path $i = i_0 \to i_1 \to \dots \to i_m = j$ in the variable dependency graph such that $m > 0$ and for any $l \in \{ 1, \dots, m - 1 \}$, we have $i_l > k$.
	For each $i = 1, \dots, n$, the \emph{alternation depth} $\alpha_i$ of the $i$-th equation is defined as follows, where $[\eta_j \neq \eta_i]$ is $1$ if $\eta_j \neq \eta_i$ and $0$ otherwise.
	\[ \alpha_i \quad\coloneqq\quad \max\ \{ 1 \} \cup \{ \alpha_j + [\eta_j \neq \eta_i] \mid j > i \land i \xrightarrow[> i]{}^{+} j \xrightarrow[> i]{}^{+} i \} \]
\end{definition}

It immediately follows from the definition that we have $\alpha_i \le i$.

\begin{example}\label{ex:alternation-depth}
	Suppose that we have the following $\mu$CLP.
	\[ X_1(x) =_{\mu} x \ge 0 \land X_3(x - 1);\qquad X_2(x) =_{\nu} X_1(x);\qquad X_3(x) =_{\mu} X_2(x) \]
	Then, its variable dependency graph $(V, F, \eta)$ is given as follows.
	\[ V = \{ 1, 2, 3 \}, \qquad F = \{ (1, 3), (2, 1), (3, 2) \}, \qquad (\eta_1, \eta_2, \eta_3) = (\mu, \nu, \mu) \]
	In this case, the alternation depths are given by $\alpha_1 = \alpha_2 = \alpha_3 = 1$.
\end{example}

\begin{example}
	If the variable dependency graph of a $\mu$CLP $E = \{ X_1 =_{\eta_1} \phi_1; \dots; X_n =_{\eta_n} \phi_n \}$ is a complete graph, then the alternation depth is determined by the number of alternations of $\mu$ and $\nu$ in the sequence of equations.
	For example, if $n = 5$ and $(\eta_1, \eta_2, \eta_3, \eta_4, \eta_5) = (\mu, \mu, \nu, \nu, \mu)$, then we have $(\alpha_1, \alpha_2, \alpha_3, \alpha_4, \alpha_5) = (3, 3, 2, 2, 1)$.
\end{example}

\begin{definition}\label{def:priority-reassignment-by-alternation-depth}
	Let $\alpha_{\max} \coloneqq \max \{ \alpha_k \mid k = 1, \dots, n \}$ and $\sigma : \Range{1}{n} \to \Range{1}{\alpha_{\max} + 1}$ be a function defined as follows.
	\[ \sigma(i) \ \coloneqq\ \begin{cases}
		\text{the largest odd number less than or equal to $2 + \alpha_{\max} - \alpha_i$} & \eta_i = \mu \\
		\text{the largest even number less than or equal to $2 + \alpha_{\max} - \alpha_i$} & \eta_i = \nu
	\end{cases} \]
\end{definition}

\begin{lemma}\label{lem:alternation-depth-parity-condition}
	Let $w = w_0 w_1 \dots$ be an infinite walk in the variable dependency graph (i.e., an infinite sequence such that $(w_i, w_{i + 1}) \in E$ for each $i$).
	The walk $w$ satisfies the parity condition $(\mathrm{id}, I_{\nu})$ if and only if $w$ satisfies the parity condition $(\sigma, \mathbf{Even}_{\le \alpha_{\max} + 1})$ where $\mathrm{id} : V \to \Range{1}{n}$ is the identity priority function, $\sigma$ is as in Definition~\ref{def:priority-reassignment-by-alternation-depth}, $\mathbf{Even}_{\le \alpha_{\max} + 1} = \{ 2, 4, \dots, 2 \lceil\alpha_{\max} / 2 \rceil \}$ is the set of even numbers up to $\alpha_{\max} + 1$, and $I_{\nu} = \{ k \mid \eta_k = \nu \}$.
\end{lemma}
\begin{proof}
	Let $m = \min \mathrm{Inf}_{\mathrm{id}}(w)$.
	Let $i_1$ be an index such that $w_{i_1} = m$ and for any $i \ge i_1$, we have $w_i \ge m$.
	Let $i_1 < i_2 < \dots$ be the infinite sequence of all the indices after $i_1$ such that $w_{i_k} = m$.
	Then, for any $k$ and any $l$ such that $i_k < l < i_{k + 1}$, we have $w_l > m$ and $w_{i_k} \xrightarrow[> m]{}^{+} w_l \xrightarrow[> m]{}^{+} w_{i_{k + 1}}$.
	By definition of alternation depth, we have $\alpha_{w_{i_k}} > \alpha_{w_l}$ if $\eta_m \neq \eta_{w_l}$ and $\alpha_{w_{i_k}} \ge \alpha_{w_l}$ if $\eta_m = \eta_{w_l}$.
	In either case, we have $\sigma(m) = \sigma(w_{i_k}) \le \sigma(w_l)$ and thus $\sigma(m) = \min \mathrm{Inf}_{\sigma}(w)$.
	Now, the claim follows because we have $\eta_m = \nu$ if and only if $\sigma(m)$ is even by definition of $\sigma$.
\end{proof}

Lemma~\ref{lem:alternation-depth-parity-condition} allows us to reassign priorities to the equations via $\sigma$ and gives a more compact representation of the parity relation~\eqref{eq:on-exit-progress-measure-relation-disjoint-union} as follows.

\begin{definition}\label{def:parity-relation-implementation-naive}
	Let $\{ r_i : \mathcal{D}_i \rightharpoonup \mathbf{Ord}^{\alpha_{\max} + 1} \}_{i = 1, \dots, n}$ be a family of partial functions.
	We define a parity relation $R = \{ R_{i, j} \}_{i, j = 1, \dots, n}$ as follows: if $(i, j)$ is an edge of the variable dependency graph, then
	\[ R_{i, j}(x, y)\ \coloneqq\ \begin{cases}
		r_i(x) {\downarrow} \land r_j(y) {\downarrow} \land r_i(x) >_{\sigma(i)} r_j(y) & \text{if $\eta_i = \mu$ ($\sigma(i)$ is odd)} \\
		r_i(x) {\downarrow} \land r_j(y) {\downarrow} \land r_i(x) \ge_{\sigma(i)} r_j(y) & \text{if $\eta_i = \nu$ ($\sigma(i)$ is even)}
	\end{cases} \]
	and $R_{i, j}(x, y) \coloneqq \False$ otherwise.
\end{definition}

By Lemma~\ref{lem:alternation-depth-parity-condition}, the relation $R$ defined in Definition~\ref{def:parity-relation-implementation-naive} is indeed a parity relation with respect to the parity condition $(p, I_{\nu})$ in Example~\ref{ex:parity-condition-muclp}, where $p : \mathcal{D} \to \{ 1, \dots, n \}$ maps elements in $\mathcal{D}_i$ to $i$.
Let $x = x_0 x_1 \dots$ be an infinite sequence such that $R(x_k, x_{k + 1})$ for any $k$.
The progress measure $r$ ensures that $\min \mathrm{Inf} (\sigma(p(x)))$ is even, i.e., the sequence $p(x)$ satisfies the parity condition $(\sigma, \mathbf{Even}_{\le \alpha_{\max} + 1})$ (cf.~Example~\ref{ex:on-exit-progress-measure}).
By definition of $R$, the sequence $p(x) \in \{ 1, \dots, n \}^{\omega}$ form an infinite walk in the variable dependency graph.
By Lemma~\ref{lem:alternation-depth-parity-condition}, $p(x)$ also satisfies the parity condition $(\mathrm{id}, I_{\nu})$, and hence $x$ satisfies the parity condition $(p, I_{\nu})$.

The number of components of $r_i : \mathcal{D}_i \rightharpoonup \mathbf{Ord}^{\alpha_{\max} + 1}$ can be further reduced to $\lceil (\alpha_{\max} + 1) / 2 \rceil$ without losing expressiveness.
For example, suppose that we are given $\{ r_i : \mathcal{D}_i \rightharpoonup \mathbf{Ord}^{\alpha_{\max} + 1} \}_{i = 1, \dots, n}$ where $\alpha_{\max} = 4$.
Then, we construct $\{ r'_i : \mathcal{D}_i \rightharpoonup \mathbf{Ord}^{3} \}_{i = 1, \dots, n}$ as follows.
\[ r_i = (r_{i, 1}, r_{i, 2}, r_{i, 3}, r_{i, 4}, r_{i, 5}) \quad\mapsto\quad r'_i = ((r_{i, 1}, r_{i, 2}), (r_{i, 3}, r_{i, 4}), r_{i, 5}) \]
Here, we assume that a pair of ordinals $(\alpha, \beta) \in \mathbf{Ord}^2$ ordered lexicographically is encoded as a single ordinal.
Then, for each odd number $k$, $r_i(x) >_k r_j(y)$ implies $r'_i(x) >_{\lceil k/2 \rceil} r'_j(y)$, and for each even number $k$, $r_i(x) \ge_k r_j(y)$ implies $r'_i(x) \ge_{\lceil k/2 \rceil} r'_j(y)$.
This observation leads to the following definition of parity relations.

\begin{definition}\label{def:parity-relation-implementation}
	Let $\{ r_i : \mathcal{D}_i \rightharpoonup \mathbf{Ord}^{\lceil (\alpha_{\max} + 1) / 2 \rceil} \}_{i = 1, \dots, n}$ be a family of partial functions.
	We define a parity relation $R = \{ R_{i, j} \}_{i, j = 1, \dots, n}$ as follows: if $(i, j)$ is an edge of the variable dependency graph, then
	\[ R_{i, j}(x, y)\ \coloneqq\ \begin{cases}
		r_i(x) {\downarrow} \land r_j(y) {\downarrow} \land r_i(x) >_{\lceil \sigma(i)/2 \rceil} r_j(y) & \text{if $\eta_i = \mu$ ($\sigma(i)$ is odd)} \\
		r_i(x) {\downarrow} \land r_j(y) {\downarrow} \land r_i(x) \ge_{\lceil \sigma(i)/2 \rceil} r_j(y) & \text{if $\eta_i = \nu$ ($\sigma(i)$ is even)}
	\end{cases} \]
	and $R_{i, j}(x, y) \coloneqq \False$ otherwise.
\end{definition}

\subsection{Decomposition into Strongly Connected Components}
\label{sec:scc-wise-parity-relation}

A \emph{strongly connected component} (SCC) of a directed graph is a maximal set of vertices such that each vertex is reachable from any other vertex in the set.
It is known that parity games can be solved SCC-wise~\cite{FriedmannATVA2009}.
Similarly, we can synthesize parity relations SCC-wise by decomposing the variable dependency graph into SCCs.

\begin{lemma}\label{lem:scc-wise-parity-relation}
	Let $G = (V, F, \eta)$ be the variable dependency graph of a $\mu$CLP $E$ and $\{ C_1, \dots, C_m \}$ be the set of all the SCCs of $G$.
	That is, for each $k$, $C_k \subseteq V$ is a maximal set such that for any $i, j \in C_k$, there exists a path from $i$ to $j$ in $G$.
	Suppose that we have a parity relation $R^k = \{ R^k_{i, j} \}_{i, j \in C_k}$ for $C_k$ for each $k = 1, \dots, m$.
	Then, the following relation $R = \{ R_{i, j} \}_{i, j = 1, \dots, n}$ is a parity relation for $E$: if $i, j \in C_k$ for some $k$, then we define $R_{i, j} \coloneqq R^k_{i, j}$; if $i, j$ are in different SCCs and $(i, j) \in F$, then $R_{i, j} \coloneqq \True$; otherwise, $R_{i, j} \coloneqq \False$.
	\qed
\end{lemma}
\begin{proof}
	Let $\{ x_n \in \mathcal{D}_{i_n} \}_n$ be an infinite sequence such that $(x_n, x_{n + 1}) \in R_{i_n, i_{n + 1}}$ for each $n$.
	Then, this sequence must eventually stay in some SCC $C_k$.
	Since $R^k$ is a parity relation, the sequence $\{ x_n \in \mathcal{D}_{i_n} \}_n$ satisfies the parity condition.
\end{proof}

\iflong
We can use parity relations in Lemma~\ref{lem:scc-wise-parity-relation} without losing expressiveness.
\begin{lemma}
	Let $E$ be a $\mu$CLP and $R \subseteq \mathcal{D} \times \mathcal{D}$ be a parity relation for $E$ and $\{ C_1, \dots, C_m \}$ be the set of all the SCCs of the variable dependency graph of $E$.
	For each $k = 1, \dots, m$, let $R^k = \{ R_{i, j} \subseteq \mathcal{D}_i \times \mathcal{D}_j \}_{i, j \in C_k}$ be the restriction of $R$ to $C_k$.
	Then, $R^k$ is a parity relation on $\bigsqcup_{i \in C_k} \mathcal{D}_i$.
	Moreover, if we define $R' \subseteq \mathcal{D} \times \mathcal{D}$ from $\{ R^k \}_{k = 1, \dots, m}$ as in Lemma~\ref{lem:scc-wise-parity-relation}, then any post-fixed point of $E[R]$ is also a post-fixed point of $E[R']$.
\end{lemma}
\begin{proof}
	It is obvious that $R^k$ is a parity relation since any infinite sequence in $\bigsqcup_{i \in C_k} \mathcal{D}_i$ is also an infinite sequence in the whole $\bigsqcup_{i} \mathcal{D}_i$.
	To show the latter claim, it suffices to show that for any edge $(i, j)$ in the variable dependency graph, we have $R_{i, j} \subseteq R'_{i, j}$.
	If $i, j \in C_k$ for some $k$, then we have $R_{i, j} = R'_{i, j}$ by definition.
	If $i, j$ are in different SCCs and $(i, j) \in F$, then we have $R'_{i, j} = \True$ by definition.
	If $(i, j)$ is not an edge, then we have $R'_{i, j} = \False$, but this does not make any difference of $E[R]$ and $E[R']$ since $R_{i, j}$ does not appear in $E[R]$.
\end{proof}
\fi

\subsection{Combining SCC-wise Parity Relations and Priority Reassignment}
\label{sec:combining-scc-ad}
We can combine SCC-wise parity relations (Lemma~\ref{lem:scc-wise-parity-relation}) with the priority reassignment technique (Definition~\ref{def:parity-relation-implementation}) as follows to obtain a template for efficient synthesis of parity relations.
We first decompose the variable dependency graph into its SCCs using standard graph algorithms (e.g., Tarjan's algorithm~\cite{TarjanSIAMJComput1972}).
Then, we consider the parity relation in Lemma~\ref{lem:scc-wise-parity-relation} where each $R^k$ is defined as in Definition~\ref{def:parity-relation-implementation}.
Here, the alternation depths are computed for each SCC $C_k$ from the subgraph induced by $C_k$.
Moreover, if an SCC $C_k$ consists of only $\nu$-vertices, then $R^k = \True$ is a valid choice of parity relation for $C_k$.
This allows us to further simplify the template and reduce the size of the generated constraints.

By Definition~\ref{def:alternation-depth}, the alternation depth $\alpha_i$ is not affected by the vertices outside the SCC containing $i$.
In this sense, the priority reassignment technique is not completely orthogonal to the SCC decomposition technique.
On the other hand, as we have shown in Example~\ref{ex:alternation-depth}, even if the variable dependency graph consists of a single SCC, the alternation depth can reduce the number of components of the functions $r_i$.

%% file: experiment.tex
\section{Implementation and Experiments}
\label{sec:experiments}

Following the approach described in Section~\ref{sec:implementation}, we extended \MuVal\footnote{\url{https://github.com/hiroshi-unno/coar}}, an existing solver for $\mu$CLP~\cite{UnnoPOPL2023}, to implement the new $\mu$-to-$\nu$ transformation presented in Section~\ref{sec:main-theorems}.
We also implemented the template for parity predicate variables defined in Definition~\ref{def:parity-relation-implementation}, together with the optimization techniques based on SCC decomposition described in Section~\ref{sec:scc-wise-parity-relation}.
We call our extended solver \MuValPPM.
In \PCSat, each component $r_{i, j} : \mathcal{D}_i \rightharpoonup \mathbf{Ord}$ of a function
$r_i : \mathcal{D}_i \rightharpoonup \mathbf{Ord}^{\lceil (\alpha_{\max} + 1) / 2 \rceil}$ in Definition~\ref{def:parity-relation-implementation}
is represented using a template of piecewise lexicographic affine functions, whose shape is almost identical to that used in~\cite{UnnoCAV2021}.
The template has several parameters that control its expressiveness, such as the number of pieces and the dimension of the lexicographic functions.
These parameters are shared among all functions $r_i$ that constitute a parity relation.
\PCSat\ incrementally increases the expressiveness of the template in a counterexample-guided manner (see~\cite{UnnoCAV2021} for details of the procedure).
Note that, as explained in~\cite{UnnoPOPL2023}, \MuVal\ solves both the original (primal) $\mu$CLP and its dual in parallel.

We evaluated \MuValPPM\ on two groups of benchmarks: (i) 202 benchmarks for validity checking of $\mu$CLPs taken from~\cite{UnnoPOPL2023}, which include LTL, CTL, CTL*, and modal $\mu$-calculus verification problems, and (ii) 1222 benchmarks from the Termination Competition 2025 (Integer Transition Systems).
We refer to the former group as the POPL'23 benchmarks and the latter group as the TermCOMP ITS benchmarks.
Our experiments were conducted on AWS t3.2xlarge instances equipped with an Intel(R) Xeon(R) Platinum 8259CL CPU @ 2.50 GHz and 32 GB of memory.
The timeout is set to 300 seconds.
Via these experiments, we aim to answer the following research question:
\begin{enumerate}[label=\textbf{RQ\arabic*},leftmargin=*]
	\item How does our implementation of the new $\mu$-to-$\nu$ transformation compare with the existing tools for $\mu$CLP solving, \MuVal~\cite{UnnoPOPL2023} and \MuStrat~\cite{TsukadaPOPL2025}?
	\label{item:rq-ppm}
	\item How effective are the proposed optimization techniques in Section~\ref{sec:implementation}?
	\label{item:rq-optimization}
\end{enumerate}

\begin{table}[tb]
	\caption{Number of solved instances. SCC indicates that SCC-wise parity relations are used, while AD indicates that priority reassignment based on alternation depth is used. Out-of-memory errors are counted as aborts. The top half of the table shows the results for the POPL'23 benchmarks, while the bottom half shows TermCOMP ITS benchmarks.}
	\label{tab:experiments-solved-instances}
	\small
	\begin{tabular}{l|cccc|c|c}
		& \multicolumn{4}{c|}{\MuValPPM} & \MuVal & \MuStrat \\
		& SCC \& AD & AD & SCC & none & \cite{UnnoPOPL2023} & \cite{TsukadaPOPL2025} \\
		\hline
		valid & 100 & 94 & 99 & 90 & 100 & 100 \\
		invalid & 89 & 85 & 87 & 87 & 92 & 75 \\
		timeout + abort & 13+0 & 23+0 & 16+0 & 23+2 & 10+0 & 27+0 \\
		\midrule
		valid & 535 & 464 & 534 & 439 & 557 & 429 \\
		invalid & 451 & 428 & 449 & 410 & 459 & 382 \\
		timeout + abort & 228+8 & 317+13 & 226+13 & 343+30 & 202+4 & 402+9
	\end{tabular}
\end{table}

\begin{figure}[tb]
	\centering
	\begin{minipage}[t]{0.5\linewidth}
		\centering
		\includegraphics[width=0.6\linewidth]{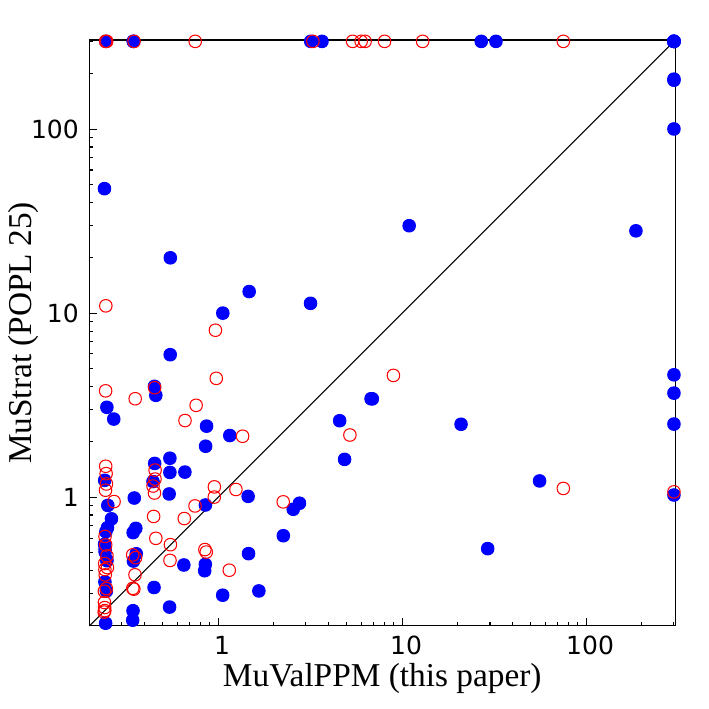}
		\includegraphics[width=0.6\linewidth]{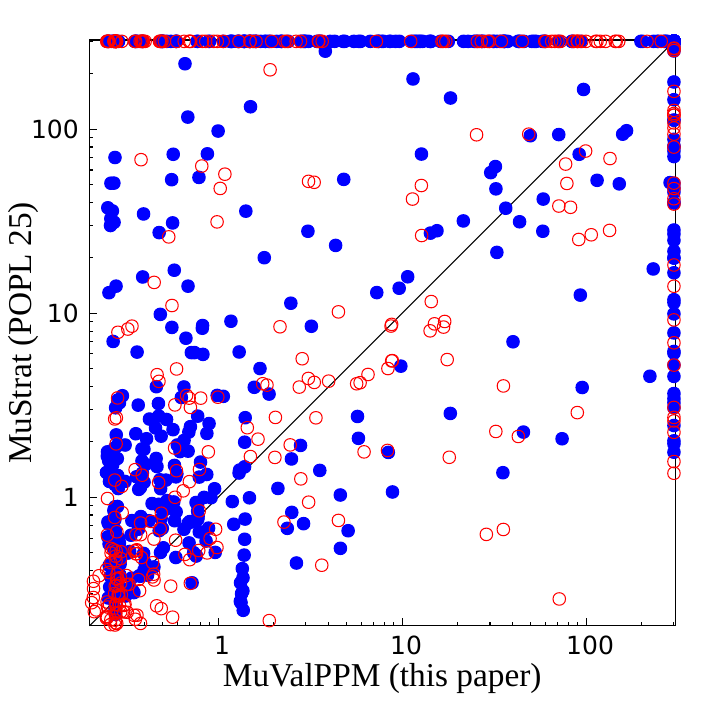}
		\subcaption{Comparison with \MuStrat~\cite{TsukadaPOPL2025}.}
		\label{subfig:pr-vs-dwf}
	\end{minipage}%
	\begin{minipage}[t]{0.5\linewidth}
		\centering
		\includegraphics[width=0.6\linewidth]{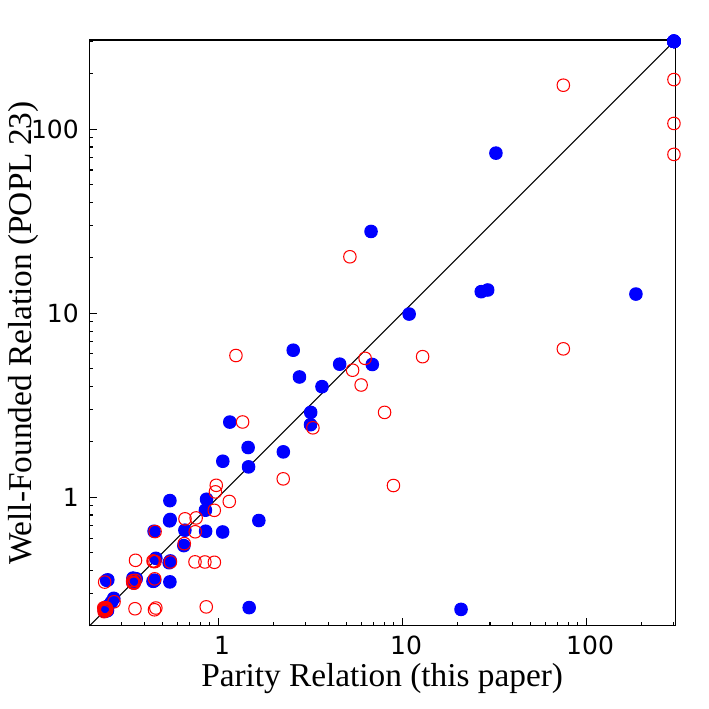}
		\includegraphics[width=0.6\linewidth]{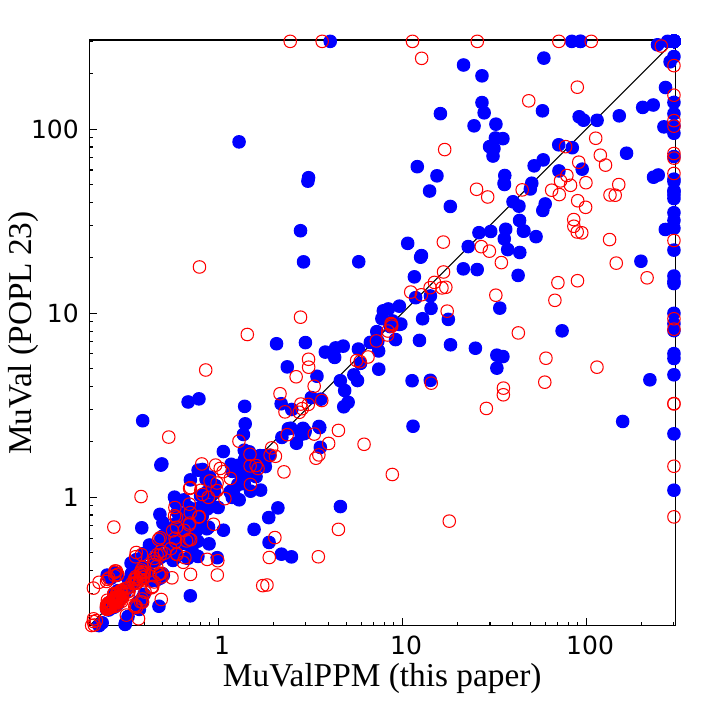}
		\subcaption{Comparison with \MuVal~\cite{UnnoPOPL2023}.}
		\label{subfig:pr-vs-wf}
	\end{minipage}
	\caption{Experimental results for \ref{item:rq-ppm}. Blue dots indicate ``valid'', while red dots indicate ``invalid''. The top row and bottom row show the results for POPL'23 benchmarks and TermCOMP ITS benchmarks, respectively.}
	\label{fig:experiments-comparison}
\end{figure}

\paragraph{Results for \ref{item:rq-ppm}.}
Fig.~\ref{fig:experiments-comparison} shows the scatter plots that compare the performance of \MuValPPM\ with that of \MuStrat~\cite{TsukadaPOPL2025} and \MuVal~\cite{UnnoPOPL2023}.
The numbers of solved instances are summarized in Table~\ref{tab:experiments-solved-instances}.

In terms of the number of solved instances, \MuValPPM\ based on parity relations performed slightly worse than \MuVal~\cite{UnnoPOPL2023}, which is based on well-founded relations, but outperformed \MuStrat~\cite{TsukadaPOPL2025}, which is based on disjunctively well-founded relations.
However, the scatter plots do not indicate that any one solver consistently dominates the others in terms of solving time.
Moreover, each solver solved instances that neither of the others could solve.
In the TermCOMP ITS benchmark suite, \MuValPPM, \MuVal, and \MuStrat\ uniquely solved 6, 14, and 39 instances, respectively.
Furthermore, a virtual best solver that always selects the fastest of the three solvers solves 587 satisfiable and 479 unsatisfiable TermCOMP ITS benchmarks, substantially more than any individual solver.
These results suggest that the three solvers are complementary and that combining their strengths appropriately could be a promising direction.

More concretely, the comparison between \MuValPPM\ and \MuVal~\cite{UnnoPOPL2023} reveals a trade-off.
On the one hand, parity relations often require more functions than well-founded relations.
On the other hand, well-founded relations require copying the arguments of predicate variables from previous visits, whereas parity relations do not.

For example, \MuVal~\cite{UnnoPOPL2023} solves the following benchmark (\texttt{lines3.hes}) by introducing only a single well-founded relation for $Y$, whereas \MuValPPM\ requires a parity relation $R={R_{i,j}}_{i,j=1,2}$ consisting of four components, resulting in a longer solving time:
\[ X(a, b) =_{\nu} a > b \land X(a + 2, b + 1) \land Y(b, a);\qquad
Y(a, b) =_{\mu} X(a, b) \lor Y(a + 2, b + 1) \]
The query for this benchmark is $\forall (m:\mathtt{int}),(n:\mathtt{int}).\ n \le m \lor X(n,m)$.

\begin{figure}[tb]
	\centering
	\begin{minipage}[t]{0.3\linewidth}
		\centering
		\includegraphics[width=\linewidth]{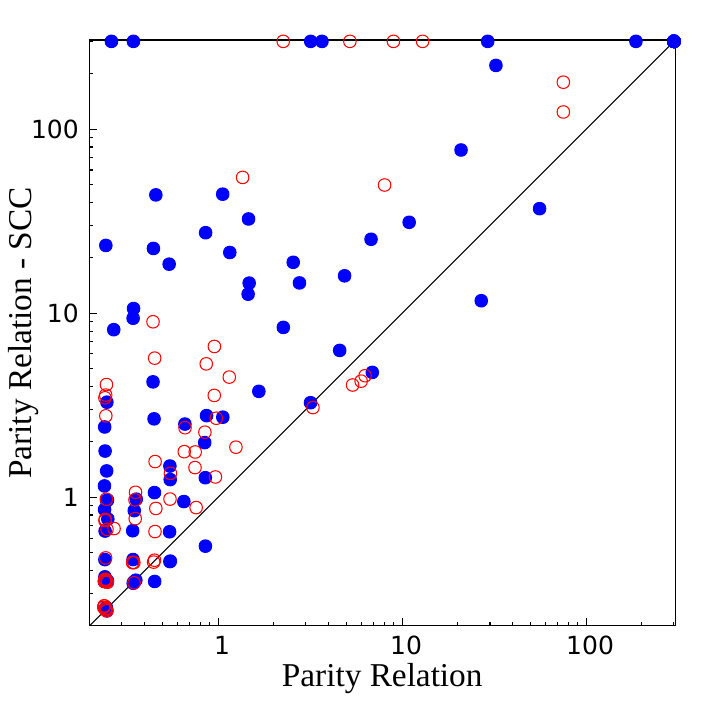}
		\includegraphics[width=\linewidth]{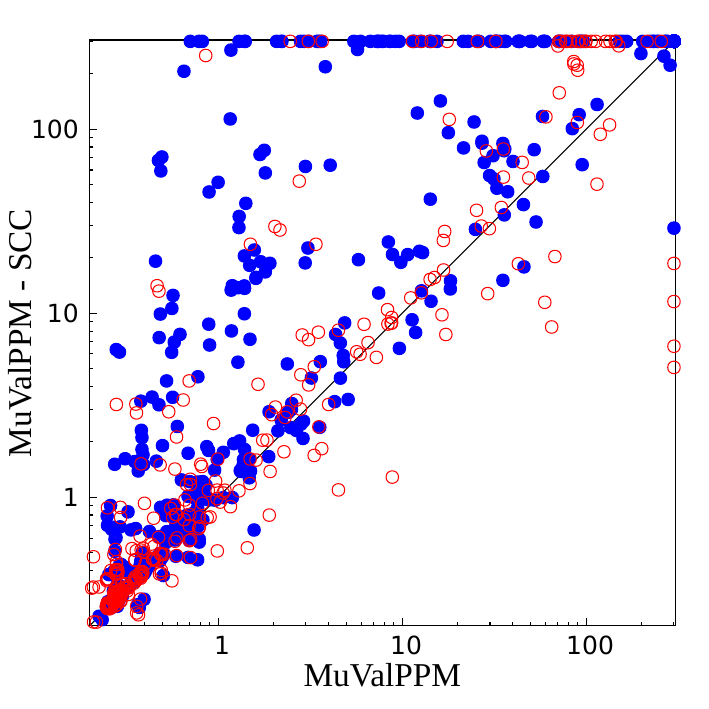}
		\subcaption{With and without SCC-wise parity relations.}
		\label{subfig:pr-vs-wo-scc}
	\end{minipage}
	\hspace{1em}
	\begin{minipage}[t]{0.3\linewidth}
		\centering
		\includegraphics[width=\linewidth]{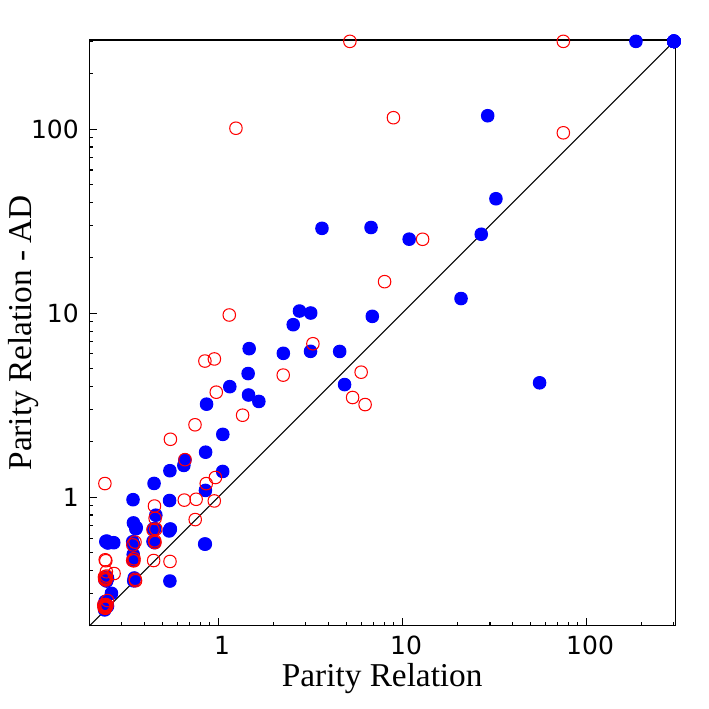}
		\includegraphics[width=\linewidth]{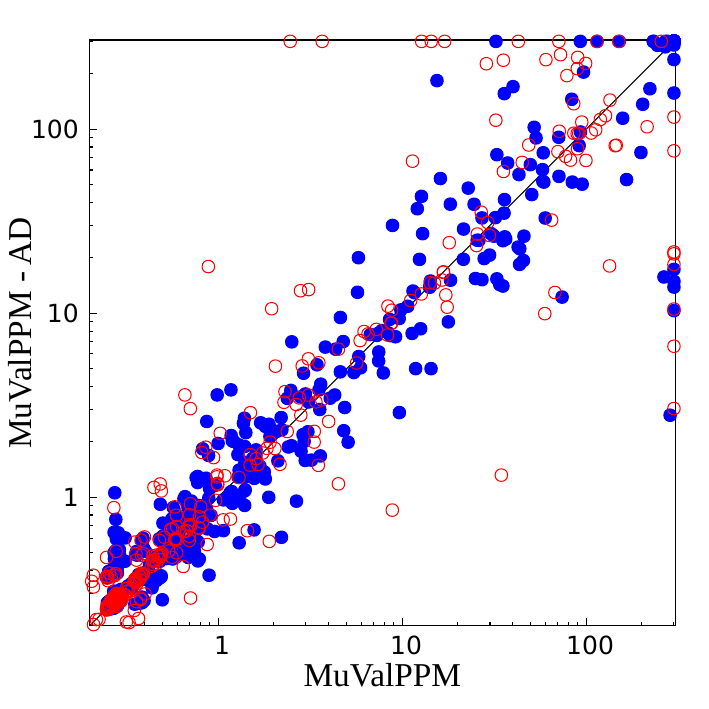}
		\subcaption{With and without priority reassignment.}
		\label{subfig:pr-vs-wo-ad}
	\end{minipage}
	\hspace{1em}
	\begin{minipage}[t]{0.3\linewidth}
		\centering
		\includegraphics[width=\linewidth]{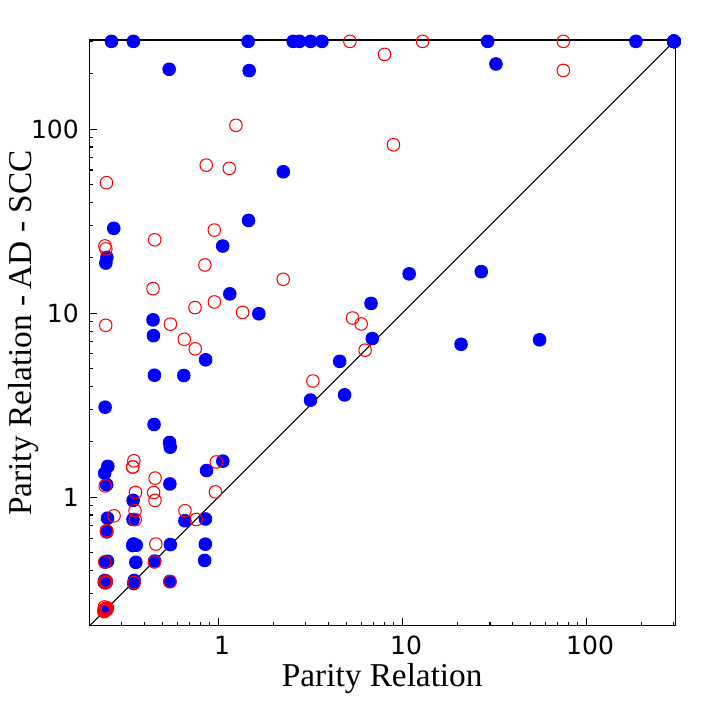}
		\includegraphics[width=\linewidth]{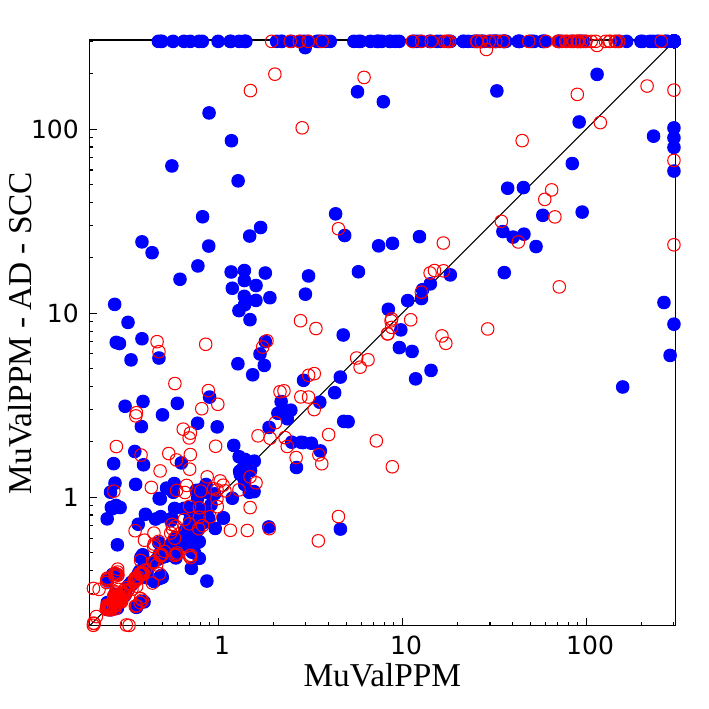}
		\subcaption{With and without SCC-wise parity relations and priority reassignment.}
		\label{subfig:pr-vs-wo-ad-scc}
	\end{minipage}
	\caption{Experimental results for \ref{item:rq-optimization}. Blue dots indicate ``valid'', while red dots indicate ``invalid''. The top row and bottom row show the results for POPL'23 benchmarks and TermCOMP ITS benchmarks, respectively.}
	\label{fig:experiments-optimization}
\end{figure}

\paragraph{Results for \ref{item:rq-optimization}.}
Fig.~\ref{fig:experiments-optimization} illustrates the effects of priority reassignment based on alternation depth (Section~\ref{sec:dependency-analysis}) and SCC-wise parity relations (Section~\ref{sec:scc-wise-parity-relation}).
Applying both optimization techniques significantly improves the performance of \MuValPPM\ (Fig.~\ref{subfig:pr-vs-wo-ad-scc}).
Figures~\ref{subfig:pr-vs-wo-scc} and~\ref{subfig:pr-vs-wo-ad} show the effect of each optimization technique in isolation.
Figure~\ref{subfig:pr-vs-wo-scc} demonstrates that SCC-wise parity relations are particularly effective, whereas Figure~\ref{subfig:pr-vs-wo-ad} shows that, once SCC-wise parity relations are applied, the additional benefit of priority reassignment is relatively modest.
This observation is consistent with the discussion in Section~\ref{sec:combining-scc-ad}: although priority reassignment has a theoretical advantage of its own, its effect partially overlaps with that of SCC-wise parity relations.

%% file: related-work.tex
\section{Related Work}
\label{sec:related-work}

\paragraph{Fixed-point logics and parity games.}
It is well-known that the model-checking problem for modal $\mu$-calculus can be reduced to solving parity games (see, e.g., \cite{BradfieldHandbookofModelChecking2018}), which can be done in quasi-polynomial time~\cite{CaludeSTOC2017}.
Various extensions of this result have been studied.
One direction is to consider fixed-point logics with more general truth values.
The reduction to parity games is extended for continuous complete lattices~\cite{BaldanPOPL2019} and for complete lattices~\cite{BaldanCONCUR2020}.
Another direction is to consider more general transition systems based on coalgebraic frameworks~\cite{HasuoPOPL2016,HausmannCONCUR2019,HausmannVMCAI2024}.

\paragraph{Transformations of first-order fixed-point logics.}
Transforming fixed-point equation systems is a common technique for solving them~\cite{KobayashiTACAS2020,NeeleLMCS2024,KobayashiSAS2019,UnnoPOPL2023}.
Among them, our focus is on $\mu$-to-$\nu$ transformations.
In \cite{KobayashiSAS2019}, a $\mu$-equation $X(\tilde{x}) =_{\mu} \phi$ is transformed to a $\nu$-equation $X'(y, \tilde{x}) =_{\nu} y > 0 \land \phi[X'(y - 1, {-})/X]$ by adding an extra argument $y$ that upper-bounds the number of unfoldings of $X$.
Then, it follows that the solution $X'^{*}$ for $X'$ satisfies $X'^{*}(y, \tilde{x}) = \interpret{\phi}^y(\bot)$.
Thus, for any term $t$, $X'^{*}(t, \tilde{x})$ is a lower approximation of $X^{*}(\tilde{x})$.
This idea was later generalized to higher-order setting in \cite{KobayashiPOPL2023}.
Another transformation is proposed in~\cite{UnnoPOPL2023}, which we have discussed in Section~\ref{sec:comparison-popl23}.
There are other kinds of transformations.
Fold/unfold transformations are studied in~\cite{KobayashiTACAS2020}, and swapping and substitution that preserve the solution are studied in~\cite{NeeleLMCS2024}.
These transformations themselves do not directly solve fixed-point equations, but they can be used to simplify the equations before solving them.

%% file: complete-progress-measure.tex
\section{Complete Progress Measures for Equation Systems}\label{sec:complete-progress-measure}

Given a $\mu$CLP $E$, there exists a progress measure $\{ r_i : \mathcal{D}_i \to \mathbf{Ord}^n \}_i$ that is \emph{complete} in the sense that for \emph{all} $\tilde{v} \in X^{*}_i$, the progress measure witnesses a winning strategy for Verifier from any position $(X_i, \tilde{v})$.
The complete progress measure is constructed from $\mu$-approximants~\cite[Section~3.4]{BaldanPOPL2019}.
Since this construction is used to prove completeness (Theorem~\ref{thm:completeness}), we will explain it in detail.

We first restrict ordinals by setting an upper bound $L \in \mathbf{Ord}^n$.
This upper bound is chosen so that we can still rely on Kleene's fixed-point theorem after the restriction.

\begin{definition}
	Let $A$ be a complete lattice and $f : A \to A$ be a monotone function.
	For any post-fixed point $x \in A$ (i.e., $x \le f(x)$) and any ordinal $\alpha$, we define the \emph{$\alpha$-th iteration} $f^{\uparrow \alpha}(x)$ by transfinite induction as follows.
	\[  f^{\uparrow 0}(x) \coloneqq x \qquad f^{\uparrow \alpha + 1}(x) \coloneqq f(f^{\uparrow \alpha}(x)) \qquad f^{\uparrow \lambda}(x) \coloneqq \sup_{\alpha < \lambda} f^{\uparrow \alpha}(x) \quad \text{if $\lambda$ is limit} \]
\end{definition}

\begin{theorem}[Kleene's fixed-point theorem]\label{thm:kleene}
	For any complete lattice $A$, there exists a sufficiently large ordinal $\alpha$ such that for any monotone function $f : A \to A$, we have $f^{\uparrow \alpha}(\bot) = \mu f$.
	\qed
\end{theorem}

Given a $\mu$CLP $E$, let $L^E = (L^E_1, \dots, L^E_n)$ be a tuple of ordinals such that if $\eta_i = \nu$, then $L^E_i = 0$, and if $\eta_i = \mu$, then $L^E_i$ is a sufficiently large ordinal such that $f^{\uparrow L^E_i}(\bot) = \mu f$ for any monotone function $f : (\mathcal{D}_i \to 2) \to (\mathcal{D}_i \to 2)$ (cf. Theorem~\ref{thm:kleene}).
We often omit the superscript $E$ when it is clear from the context.
Let $\mathbf{Ord}^i_{\le L} = \{ (l_1, \dots, l_i) \mid \forall j. l_j \le L_j \}$ be the set of tuples of ordinals that are pointwise less than or equal to $L$.
We define the order relation on $\mathbf{Ord}^i_{\le L}$ as the lexicographic order.

\begin{definition}\label{def:canonical-progress-measure}
	For each $i \in \{ 1, \dots, n \}$ and $(l_1, \dots, l_i) \in \mathbf{Ord}^i_{\le L}$, we inductively define $X_i^{(l_1, \dots, l_i)} : \mathcal{D}_i \to \{ 0, 1 \}$ as follows.
	\begin{align}
		&X_i^{(l_1, \dots, l_i)} \quad\coloneqq \\
		&\begin{cases}
			\nu X_i. \interpret{\phi_i}^{(i)}(X_1^{(l_1)}, \dots, X_{i-1}^{(l_1, \dots, l_{i-1})}, X_{i}) & \text{if $\eta_i = \nu$} \\
			\sup_{(l'_1, \dots, l'_{i - 1}) <_{\mathrm{lex}} (l_1, \dots, l_{i - 1})} X_i^{(l'_1, \dots, l'_{i - 1}, L_i)} & \text{if $\eta_i = \mu$ and $l_i = 0$} \\
			\big(\interpret{\phi_i}^{(i)}(X_1^{(l_1)}, \dots, X_{i-1}^{(l_1, \dots, l_{i-1})}, {-})\big)^{\uparrow l_i}(X_i^{(l_1, \dots, l_{i - 1}, 0)}) & \text{if $\eta_i = \mu$ and $l_i > 0$}
		\end{cases}
	\end{align}
\end{definition}
In Definition~\ref{def:canonical-progress-measure}, $l_i$ keeps track of how many times the $i$-th $\mu$-equation is unfolded and $X_i^{(l_1, \dots, l_i)}$ is the ``$(l_1, \dots, l_i)$-th'' under-approximation of the least fixed point of the $i$-th equation.
This is called a $\mu$-approximant in~\cite[Section~3.4]{BaldanPOPL2019}.
The following properties play a key role in constructing a complete progress measure.

\begin{lemma}\label{lem:canonical-progress-measure}
	\begin{enumerate}
		\item $X_i^{(l_1, \dots, l_i)}$ is monotone with respect to $(l_1, \dots, l_i) \in \mathbf{Ord}^i_{\le L}$.
		\item $(X_1^{(L_1)}, \dots, X_n^{(L_1, \dots, L_n)})$ is the solution of $E$.
		\item Let $(l_1, \dots, l_i) \in \mathbf{Ord}^i_{\le L}$ and $l' = (l_1, \dots, l_i, L_{i + 1}, \dots, L_n)$.
		We write $l' \downharpoonright i$ for the prefix of $l'$ of length $i$.
		\begin{itemize}
			\item If $\eta_i = \mu$, then
			$X_i^{(l_1, \dots, l_i + 1)} = \interpret{\phi_i}(X_1^{l' \downharpoonright 1}, \dots, X_n^{l' \downharpoonright n})$.
			\item If $\eta_i = \nu$, then
			$X_i^{(l_1, \dots, l_i)} = \interpret{\phi_i}(X_1^{l' \downharpoonright 1}, \dots, X_n^{l' \downharpoonright n})$.
		\end{itemize}
	\end{enumerate}
\end{lemma}

To prove Lemma~\ref{lem:canonical-progress-measure}, we first prove the following lemma.
\begin{lemma}\label{lem:solution-fixed-point}
	If $(X_1^{*}, \dots, X_n^{*})$ is the solution of $E$, then for any $i$ and $j$,
	\[ X_i^{*} = \interpret{\phi_i}^{(j)}(X_1^{*}, \dots, X_j^{*}). \]
\end{lemma}
\begin{proof}
	By induction on $j$.
	\begin{itemize}
		\item Base case: $j = 0$. Trivial.
		\item Step case: Suppose the claim holds for $j$.
		First note that if $i = j + 1$, then $X_{j + 1}^{*}$ is the fixed point of the following form.
		\begin{align}
			&X_{j + 1}^{*} \\
			&= \interpret{\phi_{j + 1}}^{(j)}(X_1^{*}, \dots, X_j^{*}) \\
			&= \interpret{\phi_{j + 1}}^{(j + 1)}(X_1^{*}, \dots, X_j^{*}, \eta_{j + 1} X_{j + 1}. \interpret{\phi_{j + 1}}^{(j + 1)}(X_1^{*}, \dots, X_j^{*}, X_{j + 1})) \\
			&= \eta_{j + 1} X_{j + 1}. \interpret{\phi_{j + 1}}^{(j + 1)}(X_1^{*}, \dots, X_j^{*}, X_{j + 1})
		\end{align}
		Therefore, we have the following.
		\begin{align}
			X_i^{*} &= \interpret{\phi_i}^{(j)}(X_1^{*}, \dots, X_j^{*}) \\
			&= \interpret{\phi_i}^{(j + 1)}(X_1^{*}, \dots, X_j^{*}, \eta_{j + 1} X_{j + 1}. \interpret{\phi_{j + 1}}^{(j + 1)}(X_1^{*}, \dots, X_j^{*}, X_{j + 1})) \\
			&= \interpret{\phi_{i}}^{(j + 1)}(X_1^{*}, \dots, X_j^{*}, X_{j + 1}^{*}) \tag*{\qed}
		\end{align}
	\end{itemize}
\end{proof}

\begin{proof}[Proof of Lemma~\ref{lem:canonical-progress-measure}]
	\begin{enumerate}
		\item We simultaneously prove the following statements by induction on $i$.
		\begin{itemize}
			\item If $\eta_i = \mu$, then $X_i^{(l_1, \dots, l_{i - 1}, 0)}$ is a post-fixed point (i.e., the iteration of the monotone function is well-defined).
			\[ X_i^{(l_1, \dots, l_{i - 1}, 0)} \quad\le\quad \interpret{\phi_i}^{(i)}(X_1^{(l_1)}, \dots, X_{i-1}^{(l_1, \dots, l_{i-1})}, X_i^{(l_1, \dots, l_{i - 1}, 0)}) \]
			\item $X_i^{(l_1, \dots, l_i)}$ is monotone with respect to $(l_1, \dots, l_i)$.
		\end{itemize}
		Suppose that the above statements hold for each $i' < i$.
		Each statement for $i$ is proved as follows.
		\begin{itemize}
			\item By definition of $X_i^{(l_1, \dots, l_{i - 1}, 0)}$, it suffices to show the following for any $(l'_1, \dots, l'_{i - 1}) <_{\mathrm{lex}} (l_1, \dots, l_{i - 1})$.
			\[ X_i^{(l'_1, \dots, l'_{i - 1}, L_i)} \quad\le\quad \interpret{\phi_i}^{(i)}(X_1^{(l_1)}, \dots, X_{i-1}^{(l_1, \dots, l_{i-1})}, X_i^{(l_1, \dots, l_{i - 1}, 0)}) \]
			By definition of $L_i$, $X_i^{(l'_1, \dots, l'_{i - 1}, L_i)}$ is a fixed point, and hence by the induction hypothesis about the monotonicity, we have the following.
			\begin{align}
				X_i^{(l'_1, \dots, l'_{i - 1}, L_i)}
				&= \interpret{\phi_i}^{(i)}(X_1^{(l'_1)}, \dots, X_{i-1}^{(l'_1, \dots, l'_{i-1})}, X_i^{(l'_1, \dots, l'_{i - 1}, L_i)}) \\
				&\le \interpret{\phi_i}^{(i)}(X_1^{(l_1)}, \dots, X_{i-1}^{(l_1, \dots, l_{i-1})}, X_i^{(l'_1, \dots, l'_{i - 1}, L_i)}) \\
				&\le \interpret{\phi_i}^{(i)}(X_1^{(l_1)}, \dots, X_{i-1}^{(l_1, \dots, l_{i-1})}, X_i^{(l_1, \dots, l_{i - 1}, 0)})
			\end{align}
			\item Suppose $(l_1, \dots, l_i) <_{\mathrm{lex}} (l'_1, \dots, l'_i)$.
		The case of $\eta_i = \nu$ is obvious by the induction hypothesis.
		\begin{align}
			X_i^{(l_1, \dots, l_i)}
			&= \nu X_i. \interpret{\phi_i}^{(i)}(X_1^{(l_1)}, \dots, X_{i-1}^{(l_1, \dots, l_{i-1})}, X_{i}) \\
			&\le \nu X_i. \interpret{\phi_i}^{(i)}(X_1^{(l'_1)}, \dots, X_{i-1}^{(l'_1, \dots, l'_{i-1})}, X_{i}) \\
			&= X_i^{(l'_1, \dots, l'_i)}
		\end{align}
		Next, we consider the case of $\eta_i = \mu$.
		Because $X_i^{(l_1, \dots, 0)}$ is a post-fixed point and $X_i^{(l_1, \dots, l_i)}$ is defined as the $l_i$-th iteration of a monotone function starting from $X_i^{(l_1, \dots, 0)}$, $X_i^{(l_1, \dots, l_i)}$ is monotone with respect to $l_i$.
		If $(l_1, \dots, l_{i - 1}) = (l'_1, \dots, l'_{i - 1})$ and $l_i < l'_i$, then the statement follows from the monotonicity of $X_i^{(l_1, \dots, l_i)}$ with respect to $l_i$.
		If $(l_1, \dots, l_{i - 1}) <_{\mathrm{lex}} (l'_1, \dots, l'_{i - 1})$, then without loss of generality, we assume $l'_i = 0$ and $l_i = L_i$.
		\[ (l_1, \dots, l_{i - 1}, l_i) \le_{\mathrm{lex}} (l_1, \dots, l_{i - 1}, L_i) <_{\mathrm{lex}} (l'_1, \dots, l'_{i - 1}, 0) \le_{\mathrm{lex}} (l'_1, \dots, l'_{i - 1}, l'_i) \]
		Then, we have $X_i^{(l_1, \dots, l_{i - 1}, L_i)} \le X_i^{(l'_1, \dots, l'_{i - 1}, 0)}$ by definition.
		\end{itemize}
		\item We first prove the following equation for each $i$.
		\begin{equation}
			X_i^{(l_1, \dots, l_{i - 1}, L_i)} \quad=\quad \eta_i X_i. \interpret{\phi_i}^{(i)}(X_1^{(l_1)}, \dots, X_{i - 1}^{(l_1, \dots, l_{i - 1})}, X_i)
			\label{eq:canonical-progress-measure-fixed-point-proof}
		\end{equation}
		The case of $\eta_i = \nu$ is obvious by definition.
		We prove the case of $\eta_i = \mu$ by induction on $(l_1, \dots, l_{i - 1})$ with respect to the lexicographic order.
		By definition of $X_i^{(l_1, \dots, l_{i - 1}, l_i)}$, it suffices to show
		\[ X_i^{(l_1, \dots, l_{i - 1}, 0)} \quad\le\quad \eta_i X_i. \interpret{\phi_i}^{(i)}(X_1^{(l_1)}, \dots, X_{i - 1}^{(l_1, \dots, l_{i - 1})}, X_i). \]
		By definition of $X_i^{(l_1, \dots, l_{i - 1}, 0)}$, it suffices to show the following for each $(l'_1, \dots, l'_{i - 1}) <_{\mathrm{lex}} (l_1, \dots, l_{i - 1})$, which obviously follows from the induction hypothesis.
		\[ X_i^{(l'_1, \dots, l'_{i - 1}, L_i)} \quad\le\quad \eta_i X_i. \interpret{\phi_i}^{(i)}(X_1^{(l_1)}, \dots, X_{i - 1}^{(l_1, \dots, l_{i - 1})}, X_i) \]
		Now, we prove the second statement of Lemma~\ref{lem:canonical-progress-measure} by induction on $i$.
		\begin{align}
			X_i^{(L_1, \dots, L_i)}
			&= \eta_i X_i. \interpret{\phi_i}^{(i)}(X_1^{(L_1)}, \dots, X_{i-1}^{(L_1, \dots, L_{i-1})}, X_{i}) \\
			&= \eta_i X_i. \interpret{\phi_i}^{(i)}(X_1^{*}, \dots, X_{i-1}^{*}, X_{i}) \\
			&= \interpret{\phi_i}^{(i - 1)}(X_1^{*}, \dots, X_{i-1}^{*}) \\
			&= X_i^{*} \qquad\qquad \text{(by Lemma~\ref{lem:solution-fixed-point})}
		\end{align}
		\item It suffices to show the following by induction on $j$.
	\[ i \le j \implies X_i^{(l_1, \dots, l_i + 1)} = \interpret{\phi_i}^{(j)}(X_1^{l' \downharpoonright 1}, \dots, X_j^{l' \downharpoonright j}) \]
	The base case $j = i$ is obvious by definition.
	\[ X_i^{(l_1, \dots, l_i + 1)} = \interpret{\phi_i}^{(i)}(X_1^{(l_1)}, \dots, X_i^{(l_1, \dots, l_i)}) \]
	Suppose that the statement holds for $j$.
	By definition of $l'$ and $i \le j$, we have $l' \downharpoonright j + 1 = (l_1, \dots, l_i, L_{i + 1}, \dots, L_{j + 1}) = (l' \downharpoonright j) \cdot L_{j + 1}$.
	\begin{align}
		&X_i^{(l_1, \dots, l_i + 1)} \\
		&= \interpret{\phi_i}^{(j)}(X_1^{l' \downharpoonright 1}, \dots, X_j^{l' \downharpoonright j}) \\
		&= \interpret{\phi_i}^{(j + 1)}(X_1^{l' \downharpoonright 1}, \dots, X_j^{l' \downharpoonright j}, \eta_{j + 1} X_{j + 1}. \interpret{\phi_{j + 1}}^{(j + 1)}(X_1^{l' \downharpoonright 1}, \dots, X_j^{l' \downharpoonright j}, X_{j + 1})) \\
		&= \interpret{\phi_i}^{(j + 1)}(X_1^{l' \downharpoonright 1}, \dots, X_j^{l' \downharpoonright j}, X_{j + 1}^{l' \downharpoonright j + 1}) \qquad\qquad \text{(by~\eqref{eq:canonical-progress-measure-fixed-point-proof})}
	\end{align}
	Therefore, the statement holds for $j + 1$.
	\end{enumerate}
\end{proof}

\begin{definition}
	For each $i = 1, \dots, n$, we define a partial function $r'_i : \mathcal{D}_i \rightharpoonup \mathbf{Ord}^i_{\le L}$ and $r_i : \mathcal{D}_i \rightharpoonup \mathbf{Ord}^n_{\le L}$ as
	\[ r'_i(x) \quad\coloneqq\quad \min \{ (l_1, \dots, l_i) \mid x \in X_i^{(l_1, \dots, l_i)} \} \qquad\qquad r_i(x) \quad\coloneqq\quad r'_i(x) \upharpoonright n \]
	where we regard $\min \emptyset$ as undefined, and $({-}) \upharpoonright n : \mathbf{Ord}^i_{\le L} \to \mathbf{Ord}^n_{\le L}$ is the function that extends a tuple by appending $0$, i.e., $(l_1, \dots, l_i) \upharpoonright n = (l_1, \dots, l_i, 0, \dots, 0)$.
	We call $r = \{ r_i \}_{i = 1, \dots, n}$ the \emph{canonical progress measure} for the $\mu$CLP $E$.
\end{definition}

\begin{lemma}\label{lem:canonical-progress-measure-lfp-successor}
	Let $x \in X_i^{*}$ and $(l_1, \dots, l_i) = r'_i(x)$.
	If $\eta_i = \mu$, then $l_i$ is a successor ordinal.
\end{lemma}
\begin{proof}
	By definition of $r'_i$ and $X^{(l_1, \dots, l_i)}$, if $l_i$ is either $0$ or a limit ordinal, then there exists $(l'_1, \dots, l'_i) <_{\mathrm{lex}} (l_1, \dots, l_i)$ such that $x \in X^{(l'_1, \dots, l'_i)}$, which contradicts the minimality of $(l_1, \dots, l_i)$.
\end{proof}

The canonical progress measure $r$ plays an essential role in the proof of completeness (Theorem~\ref{thm:completeness}).
\begin{proof}[Proof of Theorem~\ref{thm:completeness}]
	Let $r$ be the canonical progress measure for $E$ and $R$ be the parity relation~\eqref{eq:on-exit-progress-measure-relation} induced by $r$.
	Then, the solution $(X_1^{*}, \dots, X_n^{*})$ of $E$ is a post-fixed point of $E[R]$, which is shown as follows.
	Suppose $\tilde{v} \in X_i^{*}$.
	Then, $(l_1, \dots, l_i) = r'_i(\tilde{v})$ is defined.
	By case analysis on $\eta_i$, we can show that $\interpret{\phi_i}(X_1^{*} \cap R_{i, 1}(\tilde{v}), \dots, X_n^{*} \cap R_{i, n}(\tilde{v})) = X_i^{(l_1, \dots, l_i)}$.
	Therefore, we have $\tilde{v} \in \interpret{\phi_i}(X_1^{*} \cap R_{i, 1}(\tilde{v}), \dots, X_n^{*} \cap R_{i, n}(\tilde{v}))$.
\end{proof}

The canonical progress measure $r$ is also used to prove Theorem~\ref{thm:game-semantics}.
\begin{proof}[Proof of Theorem~\ref{thm:game-semantics}]
	It suffices to prove that (a) if $X_i^{*}(\tilde{v})$ is true, then there exists a Verifier's winning strategy from $(X_i, \tilde{v})$, and (b) if $X_i^{*}(\tilde{v})$ is false, then there exists a Falsifier's winning strategy from $(X_i, \tilde{v})$.
	For the former, let $r$ be the canonical progress measure for $E$.
	For any Verifier's position $(X_i, \tilde{v})$ such that $\tilde{v} \in X^{*}_i$, Verifier chooses the next move using Lemma~\ref{lem:canonical-progress-measure} as follows.
	\begin{itemize}
		\item If $\eta_i = \mu$ and $(l_1, \dots, l_i + 1) = r'(\tilde{v})$, then Verifier moves to $(X_1^{l' \downharpoonright 1}, \dots, X_n^{l' \downharpoonright n})$.
		\item If $\eta_i = \nu$ and $(l_1, \dots, l_i) = r'(\tilde{v})$, then Verifier moves to $(X_1^{l' \downharpoonright 1}, \dots, X_n^{l' \downharpoonright n})$.
	\end{itemize}
	Here, we define $l' = (l_1, \dots, l_i, L_{i+1}, \dots, L_n)$ in both cases.
	It is straightforward to show that this strategy is winning for Verifier if the initial position is $(X_i, \tilde{v})$ with $\tilde{v} \in X^{*}_i$.

	For the latter, let $\overline{r}$ be the canonical progress measure for the dual of $E$.
	Note that for any Verifier's position $(X_i, \tilde{v})$ such that $\tilde{v} \notin X^{*}_i$, the next position must be in $\{ (\mathcal{A}_1, \dots, \mathcal{A}_n) \mid \exists i, \mathcal{A}_i \not\subseteq X_i^{*} \}$.
	Thus, the Falsifier's strategy is defined by choosing $(X_i, \tilde{v})$ such that $\tilde{v} \in \mathcal{A}_i \setminus X_i^{*}$ and $\overline{r}_i(\tilde{v})$ is minimum with respect to the lexicographic order: $\overline{r}_i(\tilde{v}) = \min \{ \overline{r}_j(\tilde{u}) \mid j = 1, \dots, n; \tilde{u} \in \mathcal{A}_j \setminus X_j^{*} \}$.
	It is straightforward to show that this strategy is winning for Falsifier if the initial position is $(X_i, \tilde{v})$ with $\tilde{v} \notin X^{*}_i$.
\end{proof}

%% file: proofs.tex
\section{Ommitted Proofs}
\label{sec:proofs}

\begin{proof}[Proof of Theorem~\ref{thm:parity-relation-well-founded-relation}]
	We define $\hat{\mathcal{A}}_i$ as follows.
	\[ \hat{\mathcal{A}}_i \quad\coloneqq\quad \{ (\hat{x}_1, \dots, \hat{x}_{i - 1}, x_i) \mid \forall j < i. (\eta_j = \mu \land \hat{x}_j \neq \bot \implies (\hat{x}_j, x_i) \in Q_{j, i}) \land x_i \in \mathcal{A}_i \} \]
	where for each $i$ and $j$ such that $\eta_j = \mu$, if $j < i$, then we define a relation $Q_{j, i} \subseteq \mathcal{D}_j \times \mathcal{D}_i$ as
	\[ Q_{j, i} \quad\coloneqq\quad \{ (y_0, y_m) \mid m > 0 \land y_0 \mathrel{R} y_1 \mathrel{R} \cdots \mathrel{R} y_m \land y_0 \in \mathcal{D}_j \land y_m \in \mathcal{D}_i \land y_1, \dots, y_{m-1} \in \mathcal{D}_{> j} \} \]
	and if $j \ge i$, then we define $Q_{j, i} \coloneqq \emptyset$. 
	Note that $Q_{j, i}$ is defined similarly to $\hat{R}_{\mathcal{D}_{\le j}, \mathcal{D}_{< j}}$ except that we consider sequences $y_0, \dots, y_m$ that have not yet reached $\mathcal{D}_j$ after $y_0$.
	In particular, if $j < i$, $(x, y) \in Q_{j, i}$, and $(y, z) \in R$ for some $z \in \mathcal{D}_j$, then we have $(x, z) \in \hat{R}_{\mathcal{D}_{\le j}, \mathcal{D}_{< j}}$.
	We also have that if $j < i$, $(x, y) \in Q_{j, i}$, and $(y, z) \in R$ for some $z \in \mathcal{D}_{k}$ with $k > j$, then we have $(x, z) \in Q_{j, k}$. 

	Now, we show that $(\hat{\mathcal{A}}_1, \dots, \hat{\mathcal{A}}_n)$ is a post-fixed point of $\mathbf{elim}_{\mu}(E)$.
	Suppose that we have $(\hat{x}_1, \dots, \hat{x}_{i - 1}, x_i) \in \hat{\mathcal{A}}_i$.
	Then, we have $x_i \in \mathcal{A}_i$ and hence $\interpret{\phi_i}(\mathcal{A}_1 \cap R_{i, 1}(x_i), \dots, \mathcal{A}_n \cap R_{i, n}(x_i))(x_i)$ is true since $(\mathcal{A}_1, \dots, \mathcal{A}_n)$ is a post-fixed point of $E[R]$.
	By definition of $\mathbf{elim}_{\mu}(E)$, we have
	\begin{equation}
		\interpret{\phi_i[\lambda y. \psi_{i, 1}/X_1, \dots, \lambda y. \psi_{i, n}/X_n]}(\hat{\mathcal{A}}_1, \dots, \hat{\mathcal{A}}_n)(\hat{x}_1, \dots, \hat{x}_{i - 1}, x_i) \quad=\quad \interpret{\phi_i}(\hat{\mathcal{A}'}_1, \dots, \hat{\mathcal{A}'}_n)(x_i) \label{eq:post-fixed-point-elim-mu}
	\end{equation}
	where $\hat{\mathcal{A}'}_1, \dots, \hat{\mathcal{A}'}_n$ are defined as follows: if $\eta_j = \mu$,
	\[ \hat{\mathcal{A}'}_j \quad\coloneqq\quad \begin{cases}
		\{ u \mid (\hat{x}_1, \dots, \hat{x}_{i - 1}, u) \in \hat{\mathcal{A}}_j \land (\hat{x}_j \neq \bot \implies (\hat{x}_j, u) \in \hat{R}_{\mathcal{D}_{\le j}, \mathcal{D}_{< j}}) \} & j < i \\
		\{ u \mid (\hat{x}_1, \dots, \hat{x}_{i - 1}, u) \in \hat{\mathcal{A}}_i \land (x_i, u) \in \hat{R}_{\mathcal{D}_{\le j}, \mathcal{D}_{< j}} \} & j = i \\
		\{ u \mid (\hat{x}_1, \dots, \hat{x}_{i - 1}, x_i, \bot, \dots, \bot, u) \in \hat{\mathcal{A}}_j \} & j > i
	\end{cases} \]
	and if $\eta_j = \nu$, we replace $\hat{R}_{\mathcal{D}_{\le j}, \mathcal{D}_{< j}}$ with $\True$ in the above definition.
	Then, we have $\mathcal{A}_j \cap R_{i, j}(x_i) \subseteq \hat{\mathcal{A}'}_j$ for each $j$.
	\begin{itemize}
		\item If $j \le i$ and $y \in \mathcal{A}_j \cap R_{i, j}(x_i)$, then we have $(\hat{x}_1, \dots, \hat{x}_{j - 1}, y) \in \hat{\mathcal{A}}_j$ because for any $k < j$, if $\eta_k = \mu$ and $\hat{x}_k \neq \bot$, then $(\hat{x}_k, x_i) \in Q_{k, i}$ and $(x_i, y) \in R_{i, j}$ imply $(\hat{x}_k, y) \in Q_{k, j}$.
		Moreover, if $\eta_j = \mu$, then we have the following.
		\begin{itemize}
			\item If $j < i$ and $\hat{x}_j \neq \bot$, then we have $(\hat{x}_j, y) \in \hat{R}_{\mathcal{D}_{\le j}, \mathcal{D}_{< j}}$ because $(\hat{x}_j, x_i) \in Q_{j, i}$ and $(x_i, y) \in R_{i, j}$ imply $(\hat{x}_j, y) \in \hat{R}_{\mathcal{D}_{\le j}, \mathcal{D}_{< j}}$.
			\item If $j = i$, we have $(x_i, y) \in \hat{R}_{\mathcal{D}_{\le i}, \mathcal{D}_{< i}}$ because $(x_i, y) \in R_{i, i}$ and $R_{i, i} \subseteq \hat{R}_{\mathcal{D}_{\le i}, \mathcal{D}_{< i}}$.
		\end{itemize}
		\item If $j > i$ and $y \in \mathcal{A}_j \cap R_{i, j}(x_i)$, then we have $(\hat{x}_1, \dots, \hat{x}_{i - 1}, x_i, \bot, \dots, \bot, y) \in \hat{\mathcal{A}}_j$ because
		\begin{itemize}
			\item for any $k < i$, if $\eta_k = \mu$ and $\hat{x}_k \neq \bot$, then $(\hat{x}_k, x_i) \in Q_{k, i}$ and $(x_i, y) \in R_{i, j}$ imply $(\hat{x}_k, y) \in Q_{k, j}$; and
			\item for $k = i$, if $\eta_k = \mu$, then $(x_i, y) \in R_{i, j}$ implies $(x_i, y) \in Q_{i, j}$.
		\end{itemize}
	\end{itemize}
	By monotonicity of $\interpret{\phi_i}$, the left-hand side of~\eqref{eq:post-fixed-point-elim-mu} is true whenever $(\hat{x}_1, \dots, \hat{x}_{i - 1}, x_i) \in \hat{\mathcal{A}}_i$.
	Hence, $(\hat{\mathcal{A}}_1, \dots, \hat{\mathcal{A}}_n)$ is a post-fixed point of $\mathbf{elim}_{\mu}(E)$.
\end{proof}